\documentclass[letterpaper,11pt]{scrartcl}
\usepackage[margin=1.2in]{geometry}
\usepackage{lmodern}
\usepackage{natbib}
\usepackage{graphicx,color}

\usepackage[utf8]{inputenc}
\usepackage[T1]{fontenc}
\usepackage[colorlinks=true,allcolors=teal]{hyperref}

\usepackage{booktabs}
\usepackage[ruled]{algorithm2e}

\SetAlFnt{\small}
\SetAlCapFnt{\small}
\SetAlCapNameFnt{\small}
\SetAlCapHSkip{0pt}
\IncMargin{-\parindent}

\usepackage{amsmath,amssymb,amsthm}
\usepackage{mathtools}
\usepackage{cleveref}
\usepackage{nicefrac}
\usepackage{enumitem}
\usepackage{multirow}
\usepackage{thm-restate}

\usepackage{array}
\newcolumntype{C}[1]{>{\centering\arraybackslash}p{#1}}
\newcommand{\ceil}[1]{ \left\lceil #1 \right\rceil }
\newcommand{\floor}[1]{ \left\lfloor #1 \right\rfloor }
\newcommand{\A}{ \mathcal{A} }
\newcommand{\sort}{\mathcal{U}_k}

\newcommand{\CV}{V_q}
\newcommand{\E}{ \mathbb{E} }
\newcommand{\calD}{\mathcal{D}}
\renewcommand{\top}{\operatorname{\mathsf{top}}}
\newcommand{\hP}{\hat{P}}
\usepackage{tikz}
\DeclareMathOperator*{\argmax}{arg\,max}
\DeclareMathOperator*{\argmin}{arg\,min}
\usepackage{cleveref}
\newcommand{\fgc}{\textsc{FairGreedyCapture}}
\newcommand{\afgc}{\textsc{Augmented-FairGreedyCapture}}

\newtheorem{theorem}{Theorem}

\newtheorem{lemma}{Lemma}
\newtheorem{proposition}[theorem]{Proposition}
\newtheorem{definition}[theorem]{Definition}

\newtheorem{example}[theorem]{Example}
\theoremstyle{remark}

\usepackage{thm-restate}

\usepackage{subcaption}
\usepackage{bbm}
\usepackage{nicefrac}

\newcommand{\unifsel}{\textsc{Uniform}}
\newcommand{\adult}{\textrm{Adult}}
\newcommand{\ess}{\textrm{ESS}}

\usepackage{multirow}

\renewcommand{\cite}{\citep}

\setcitestyle{authoryear}

\title{Boosting Sortition via Proportional Representation}
\author{\textbf{Soroush Ebadian$^1$, Evi Micha$^{2}$}}
\date{\large $^1$University of Toronto\hspace{.3cm}$^2$Harvard University}

\begin{document}

\maketitle

\begin{abstract}
Sortition  is based on the idea of choosing randomly selected representatives for decision making. The main properties that make sortition particularly appealing are {\em fairness} --- all the citizens can be selected with the same probability--- and {\em proportional representation} --- a randomly selected panel  probably   reflects the composition of the whole population.  When a population lies on a representation metric, we formally define proportional representation  by using a  notion  called the {\em core}. A panel is in the core if no group of individuals is underrepresented proportional to its size. While uniform selection is fair, it does not always return panels that are in the core. Thus, we ask if we can design a selection algorithm that satisfies fairness and  {\em ex post core} simultaneously.  We answer this question affirmatively and present an efficient selection algorithm that is fair and provides a constant-factor approximation to the optimal ex post core. Moreover, we show that uniformly random selection  satisfies a constant-factor approximation to the optimal {\em ex ante core}.   We complement our theoretical results by conducting experiments with real data.
\end{abstract}

\section{Introduction}

In the last centuries, representative democracy has become synonymous with elections. However, this has not been the case throughout history. Since ancient Athens, the random selection of representatives from a given population has been proposed as a means of promoting democracy and equality~~\cite{Rey16}. Sortition has gained significant popularity in recent years, mainly because of its use for forming {\em citizens' assemblies}, where a randomly selected panel of individuals deliberates on issues and makes recommendations. Currently, citizens' assemblies are being implemented by more than 40 organizations in over 25 countries~\cite{FGGH+21}.

Recently, there has been a growing interest within the computer science research community in designing algorithms that select representative panels fairly and transparently~\cite{FGGP20,FGGH+21,FKP21,SKMPS22}. Admittedly, a straightforward method for selecting a representative panel of size $k$ from a given population of size $n$ is to randomly select $k$ individuals uniformly~\cite{Eng89}. We refer to this simple procedure as uniform selection. As highlighted by~\citet{FGGP20}, two main reasons make this method particularly appealing:

\begin{enumerate}
    \item  {\em Fairness}: Each citizen is included in the panel with the same probability, satisfying the requirement of equal participation. Specifically, each citizen is selected with a probability of $k/n$
    
    \item {\em Proportional Representation}: The selected panel is likely to mirror the structure of the population, since if $x\%$ of the population has specific characteristics, then in expectation, $x\%$ of the panel will consist of individuals with these characteristics. For instance, if the female share of the population is $48\%$, then in expectation, $48\%$ of the panel will be females.
\end{enumerate}

Indeed, uniform selection seems to achieve   proportional representation {\em ex ante} (before the randomness is realized), since in expectation the selected panel reflects the composition of the population, especially when the  size of the panel is very large. However, one of the critiques of this sampling procedure is that with non-zero probability,  a  panel that  completely excludes certain demographic groups can be selected~\cite{Eng89}.   For example,  if the population is split evenly between college-educated and non-college-educated individuals,  there's a chance that uniform selection could result in a panel consisting solely of college-educated individuals.   To address such extreme cases, various strategies have been proposed to ensure proportional representation \textit{ex post} (after the randomness is realized)~\cite{martin1999random}. 

One common strategy is the use of {\em stratified sampling}~\cite{gkasiorowska2023sortition}. The idea is that the individuals are partitioned into disjoint groups and then a proportional number of representatives is sampled uniformly at random from each group. For example, if the population is comprised of 49\% college-educated individuals and 51\% non-college-educated individuals, then we can choose 49\% of the representatives from the first group and the remaining representatives from the other group. This idea can be extended  to ensure proportional representation across intersectional features as well. For instance, in a population characterized by the level of education and the income, we can define four groups: college-educated low-income, college-educated high-income, non-college-educated low-income, and non-college-educated high-income and then sample from each group separately. However, this approach becomes impractical when dealing with a large predefined set of features, as the number of possible groups can grow exponentially, and there may not be enough seats in the panel to represent all of them. A more general  approach, extensively used in practice, is to set quotas over individual or set of features~\cite{FGGP20,vergne2018citizens}. Similar to stratified sampling, when aiming for proportional representation across all intersectional features, the number of quotas can become exponential, making it infeasible to satisfy all of them concurrently. Alternatively, one may opt for setting quotas over a subset of intersectional features. For instance, quotas could be set for gender and race simultaneously, along with additional quotas for income. However, this might not ensure the representation of specific subgroups, such as high-income black women.

  The presence of the above challenges in existing strategies prompts a need for alternative approaches for ensuring proportional representation.  This, in turn, highlights the necessity of rigorously defining proportional representation first. Our work departs from these observations, and we aim to address the following questions:
    \begin{enumerate}
        \item  {\em What is a formal definition of proportional representation of a population? }
        \item {\em To what extent does uniform selection satisfy proportional representation?}
        \item {\em Is it possible to design selection algorithms that enhance representation guarantees while maintaining fairness? }
    \end{enumerate}

\subsection{Our approach}

\paragraph{Proportional Representation via Core.}  We  begin by tackling the first question posed above. Intuitively, a panel can be deemed proportionally representative if each group of size $s$ within a population of $n$ individuals is represented by $s/n \cdot k$ members in the panel, out of the total $k$ representatives selected. Motivated by this intuition,  we   borrow  a notion of proportional representation used by   recent  works on multiwinner elections, fair allocation of public goods and clustering~\cite{aziz2017justified, fain2018fair, CFSV19,cheng2020group,chen2019proportionally}, called the {\em core}.  The main idea of the core is: {\em  Every subset $S$ of the population  is entitled to choose up to $|S|/n \cdot k$ representatives}.   Formally,
a panel $P$ is called proportionally representative, or is said to be in the core,  if  there does not exist a subset   $S$  of the population that could choose a  panel $P'$, with  $|P'|\leq |S|/n\cdot k$, under which all of them feel more represented.  Note that this notion is not defined  over predefined groups using particular features, but it provides proportional representation in the panel to {\em every } subset of the population.

\paragraph{Representation Metric Space.}   A conceptual challenge is to quantify the extent to which a panel represents  an individual. To address this, we use the same approach as taken by \citet{SKMPS22} in which it is  assumed that the individuals lie in an underlying {\em representation metric space}.  The representation metric space can be constructed as a function of features that are of particular interest for an application  at hand, such as   gender, age,
ethnicity and education. 
Intuitively, the construction of such a metric space eliminates the necessity of partitioning individuals into groups that all share exactly the same characteristics. Instead, it serves as a means of detecting large  groups of individuals that share similar characteristics and are eligible to be represented proportionally. For example, a 30-year-old single, low-income black woman might still feel close to  a 35-year-old married, medium-income black woman, since they share many characteristics, even if they differ in some of them.

\paragraph{$q$-Cost.}  Finally, to measure the degree to which an individual is represented by a panel again, we take the  approach of \citet{SKMPS22}, following a recent work of \citet{CNV22} in multiwinner elections.
Specifically, the cost of an individual for a panel is determined by her distance from the $q$-th closest member in the panel, for some $q\in [k]$. We find this choice of cost suitable for applications related to sortition due to two main reasons. First, an individual may not care about her distance to all the representatives, but she may wish to ensure that there are a few with whom she can relate. For example, a woman may want to ensure that there are at least a few women on a panel to represent her, without necessarily requiring the entire panel to be composed of females, which would not be reasonable. Second, it effectively differentiates between panels containing representatives whom an individual can readily relate to and panels where representatives are more distant from her. For instance, consider an individual aged 40,  a panel that includes 2 representatives aged 40, one representative aged 20, and one representative aged 60 and another panel consisting of two representatives aged 30 and two representatives aged 50. The individual may feel represented by at least two people in the former panel,  and   therefore for $q=2$ her cost would be low. While for the second panel her cost would be higher since no representative is that close to her. In contrast, natural alternatives such as the average distance would fail to capture this difference since both panels would have the same average distance from her. The choice of $q$ depends on the application at hand. However, in this work, we provide selection algorithms that do not require knowledge of the value of $q$ but  offer guarantees for any value of it concurrently.

\subsection{Our Contribution}
Our primary conceptual contribution lies in introducing the core,  in the context of sortition. Before delving into our work, we discuss the relevant literature that has provided inspiration and insights for our research.
The idea of using the core as a means of measuring the proportional representation that a panel provides to a population, lying in a metric space, was first introduced by~\citet{chen2019proportionally} in a clustering setting. 
In our terms, \citet{chen2019proportionally} consider the case of $q = 1$, i.e., each individual cares for her distance from her closest representative, while in this work, we extend the notion of core to the class of $q$-cost functions.
They show that a solution in the core is not guaranteed to exist and define a multiplicative approximation of it with respect to the cost improvement of all individuals eligible to choose a different panel. They introduce an algorithm, called Greedy Capture, that returns a solution in the $(1+\sqrt{2})$-approximate core. Roughly speaking, the algorithm partitions the $n$ individuals into $k$ parts by smoothly increasing balls in the underlying metric space around each individual and greedily creating a part whenever a ball captures $n/k$ individuals that have not already been captured. The centers of the balls serve as the representatives.

In a sortition setting,in addition to proportional representation of all groups, it is important to ensure the {\em fairness constraint} which is that all individuals have  the same chance of being included in the panel. For ensuring that, a selection algorithm should return distribution over panels of size $k$, and not a deterministic panel as in the clustering setting.  Therefore, in this work we ask for selection algorithms that are simultaneously in the {\em ex post core}, meaning that  {\em every }   panel that the algorithm might return, is in the core, and simultaneously  is fair, meaning that each individual is included in the panel with probability equal to $k/n$.

In~\Cref{sec:fair-greedy-capture}, as one would expect, we demonstrate that uniform selection, despite satisfying fairness by its definition, falls short of achieving any reasonable approximation to the ex post core for almost any $q$, with only exception being $q=k$. This is  due to the fact that when $q=k$, any panel inherently belongs to the $2$-approximate ex post core, as we will show  later. We then pose the question: Is there any selection algorithm that is fair and achieves an $O(1)$-approximation to the ex post core? The answer is affirmative. We introduce an efficient selection algorithm, denoted as $\fgc$, that is fair and is in the $6$-approximate ex post core for \textit{every} value of $q \in [k]$. In some sense, this guarantees the best of both worlds, as we provide an algorithm that preserves the positive characteristic of uniform selection, namely fairness, and additionally, it ensures that any realized panel is in the $O(1)$-approximate ex post core. Again loosely speaking, 
 $\fgc$  creates $k$ parts using  Greedy Capture which  ``opens'' a ball in the metric space when a  sufficiently number of   individuals fall into it.  In contrast to Greedy Capture, which selects the center of the ball as a representative, $\fgc$ assigns probabilities of selection to individuals within the ball, ensuring that the sum of these probabilities equals to $1$.  This ensures the selection of one representative from each ball. Additionally, to ensure fairness, a total fraction of $k/n$ is assigned to each individual across the $k$ balls. Then, leveraging Birkhoff’s decomposition algorithm, we find a distribution over panels of size $k$, where each panel contains at least one representative from each ball, and each individual is selected with a probability of $k/n$. 
 We complement this result by showing that no fair selection  algorithm provides an approximation better than $2$ to the ex post core.

In~\Cref{sec:uniform-selection}, we turn our attention to the question: Is uniform selection in the ex ante core? 
As previously mentioned, uniform selection seems to satisfy the ex ante core, at least for large panels, since, in expectation, a panel is proportionally representative. Here, we investigate whether this is true for all values of $k$ and $q$. In particular, we define a selection algorithm to be in the {\em ex ante core} if, for any panel $P$, the expected number of individuals who feel more represented by $P$ than panels chosen from the selection algorithms is less than $|P|/n \cdot k$. This indicates that no other panel receives significant support, in expectation. First, we show that for $q=k$, uniform selection is in the ex ante core. However, for $q < k$, no fair selection algorithm is in the ex ante core. Therefore, as before, we define a multiplicative approximation with respect to the cost improvement. We demonstrate that uniform selection provides an approximation of $4$ to the ex ante core. On the other hand, we show that no fair selection algorithm provides an approximation better than $2$ to the ex ante core.

In \Cref{sec:audit}, we explore the question of whether, given a panel $P$, there is any way to determine if it satisfies an approximation of the ex post core for a  value of $q$. This  can be useful when a panel has been sampled using a selection algorithm that does not provide any guarantees for the ex post core. We  show that given a panel $P$, we can  approximate, in polynomial time, how much it violates the core up to constants.

Finally, in~\Cref{sec:experiments}, we empirically evaluate the approximation of uniform selection and $\fgc$  to the ex post core  on constructed metrics derived from two demographic datasets. We notice that for large values of $q$,  uniform selection achieves an approximation to the ex post core similar to the one that $\fgc$  achieves. For smaller values of $q$, when the individuals form  cohesive parts, uniform selection has unbounded approximation very often. However, when the individuals are well spread in the space,  uniform selection  achieves 
a good approximation of  the ex post core. Thus, the decision of using uniform selection depends on the value of $q$ and the structure of the population. 

\subsection{Related Work}

\citet{SKMPS22} recently considered the same question of measuring the representation that a panel or a selection algorithm achieves in a rigorous way. As we mentioned above, they also assume the existence of a representative metric space and use the distance of the $q$-th closest representative in the panel to measure to what degree a panel represents an individual. However, they use the social cost (i.e. the sum of individual costs) to measure how much a panel represents the whole population. In~\Cref{app:social-cost}, we show that this measure of representation may fail to achieve the idea of proportional representation. Moreover, while a reasonable approximation of their notion of representation is, in some cases, incompatible with fairness (i.e., each individual is included in the panel with the same probability), in this work, we show that there are selection algorithms that achieve a constant approximation of proportional representation and fairness simultaneously.  

As we discussed above,  a method that is used in practice for enforcing representation is by setting quotas over features.  However, a problem that appears  is that only a few people volunteer to participate in a decision panel. As a result, the   representatives are selected from a pool of volunteers which usually does not reflect the composition of the population, since for example highly educated people are usually more willing to participate in a decision panel than less educated people. \citet{FGGH+21} proposed selection algorithms that, given a biased pool of volunteers, find distributions that   maximize the minimum selection probability of any volunteer over panels that satisfy the desired quotas. 
In this work,  similar to \citet{SKMPS22} and \citet{BGP19}, we focus on the pivotal idea of a sortition based democracy that relies on sampling representatives directly from the underlying population~\cite{GW19}.  However, later, we discuss how our approach can be modified for being applied in biased pools of volunteers.  \citet{BGP19}   focused on the idea of  stratified sampling and asked how this strategy may affect the variance of the representation of unknown groups.    \citet{FKP21} studied how the selection algorithms can become transparent as well. In a more recent work, \citet{Flanigan-strategic-2024} studied the manipulability of different selection algorithms, i.e the incentives of individuals to misreport their features.

The representation of individuals as having an ideal point in a  metric space has its roots to the spatial model of voting~\cite{Arr90, Ene84}. As we mentioned above, the idea of using the core as a notion of proportional representation in a metric space  was first introduced by~\citet{chen2019proportionally}, and later revisited by~\citet{micha2020proportionally}, in a clustering setting. Proportional representation in clustering has also been studied by \citet{aziz2023proportionally} and \citet{Kalayci2024}. The definition by \citet{aziz2023proportionally} is quite similar to the core, with the basic difference being that each dense group explicitly requires a sufficient number of representatives.
\citet{Kalayci2024} consider a version of the core where an agent's cost for the panel is the sum of the distance of each representative, and  a group is incentivized to deviate to another solution if the overall group can reduce the sum of costs.
A drawback of both the definition of the core we use in this paper and Greedy Capture, which was mentioned by \citet{aziz2023proportionally} and \citet{Kalayci2024}, is that a dense group might end up being represented by just one individual. This happens because Greedy Capture keeps expanding opened balls, and when a new individual is captured by such a ball, it disregards it by implicitly assuming that this individual is already represented. We stress that while our notion of the core does not explicitly account for this problem, $\fgc$ does not expand balls that are already open, and thus, it does not suffer from this weakness.
More broadly, the implicitly goal of clustering is to find a set of $k$ centers that represent all the data points in an underlying metric space. As discussed by \citet{chen2019proportionally} in their work, the classic objectives, namely $k$-center, $k$-means, and $k$-median objectives, are deemed incompatible with the core. Consequently, they do not align with the notion of proportional representation desired in this work. The literature has explored various notions of fairness in clustering~\cite{chhabra2021overview}. Recently, \citet{kellerhals2023proportional} establish links among the numerous concepts related to fairness and proportionality in clustering.

Proportional representation through core  has been extensively studied in the context of multiwinner elections as well~\cite{aziz2017justified, faliszewski2017multiwinner, lackner2023multi, fain2018fair}. 
The problem of selecting a representative panel can be framed as a committee election problem, where the candidates are drawn from the same pool as the voters.  While in these works, the voters and the candidates do not lie in a metric space, but instead the  voters hold rankings over candidates, in our model, the rankings could derive from the underlying metric space.  
Due to impossibility results~\cite{cheng2020group}, relaxations of the core have been studied. The ex ante core, as defined here, was introduced by~\citet{cheng2020group}. They show that, without the fairness constraint,  the ex ante core can be guaranteed. In this work,  we show that by imposing this fairness constraint, an approximation to the ex ante $q$-core  better than $2$ is impossible, for all $q\in [k-1]$.

\section{Preliminaries}

For $t \in \mathbb{N}$, let $[t]=\{1,\ldots, t\}$. We denote the population by $[n]$. A panel $P$ is defined as a subset of the population. The $n$ individuals lie in an underlying {\em representation metric space} with distance function $d$. The distance between individuals $i$ and $j$ is denoted as $d(i,j)$. We assume that the distances are symmetric, i.e., $d(i,j)=d(j,i)$, and satisfy the triangle inequality, i.e., $d(i,j)\leq d(i,\ell)+d(\ell,j)$. An instance of our problem is characterized by the individuals in the population and the distances among them. Henceforth, we simply refer to such an instance as $d$.

 We consider a class of cost functions to measure the cost of an individual $i$ within a panel $P$. For $q \in [k]$, we define the \emph{$q$-cost} of $i$ for $P$ as the distance to her $q$-th closest member in the panel, denoted by $c_{q}(i,P;d)$. When $q=1$, the cost of an individual is equal to her distance from her closest representative in the panel, and for $q=k$, the cost is equal to her distance from her furthest representative in the panel. We denote by $\top_q(i,P;d)$ the set of the $q$ closest representatives of $i$ in a panel $P$ (with ties broken arbitrarily). Additionally, $B(i,r;d)$ represents the set of individuals captured from a ball centered at $i$ with a radius of $r$, i.e., $B(i,r;d)=\{i'\in [n]: d(i,i')\leq r\}$. We may omit $d$ from the notation when clear from the context.

A {\em selection algorithm}, denoted by  $\A_{k}$, is   parameterized by $k$ and  takes as input the metric $d$ and outputs a distribution over all panels of size $k$.  We say that  a panel is  in the support of $\A_{k}$, if it is implemented with positive probability under the distribution that $\A_{k}$ outputs. We pay special attention to the \emph{uniform selection} algorithm, denoted by $\sort$, that always outputs a uniform distribution over all the subsets of the population of size $k$.

\textbf{Fairness.}~
As mentioned  above, one of the appealing properties of uniform selection is that each individual is included in the panel with the same probability. We call this property {\em fairness} and we say that a selection algorithm   is  {\em fair}  if:
\begin{align*}
  \forall i \in [n],  \quad \Pr\nolimits_{P \sim\A_{k}}[i\in P ]=k/n.
\end{align*}

\textbf{Core.}~
Another appealing  property of sortition is  proportional representation. Here, we utilize  the idea of the core to measure the proportional representation of a panel and, by extension, of a selection algorithm. To do so, we first introduce the following definition: 
 For $\alpha \ge 1$, the {\em $\alpha$-$q$-preference count} of $P$ with respect to $P'$ is the number of individuals  whose $q$-cost under $P$ is larger than $\alpha$ times their $q$-cost under $P'$:
\begin{align*}
    \CV( P,P',\alpha)= \lvert \{i\in [n] : c_q(i,P)>\alpha \cdot c_q(i,P')\} \rvert. 
\end{align*}
A panel $P$ is {\em in the $\alpha$-$q$-core}, if for any panel $P'$, $ \CV(P,P', \alpha)<  |P'|\cdot n/k$. For $\alpha=1$, we say that the panel is in the $q$-core. 
We define $\alpha$-$q$-core for $\alpha>1$, since  even when $q=1$, a panel in the exact $q$-core is not guaranteed to exist~\citep{chen2019proportionally,micha2020proportionally}.

\paragraph{Ex Post $q$-Core.}
A selection algorithm $\A_{k}$ is in the {\em ex post  $\alpha$-$q$-core} (or  ex post $q$-core, for $\alpha=1$) if every panel $P$ in the support of  $\A_{k}$ is in the $\alpha$-$q$-core,  i.e., for all  $P$ drawn from $\A_{k}$ and all $P'$,
\[
\CV(P,P', \alpha)<  |P'|\cdot n/k.
\]

\paragraph{Ex Ante $q$-Core.}
A selection algorithm $\A_{k}$ is in the  {\em ex ante $\alpha$-$q$-core} (or  ex ante $q$-core, for  $\alpha=1$) if  for all $P'$:
\begin{align*}
    \E_{P \sim \A_{k}}[\CV(P,P',\alpha)]<  |P'|\cdot \frac{n}{k}.
\end{align*}

The idea of requiring a core-like property over the expected number of preference counts was introduced by~\citet{cheng2020group} in a multi-winner election setting. Essentially, it states that for any panel $P'$, if, for any realized panel $P$, we count the number of individuals that reduce their cost by a multiplicative factor of at least $\alpha$ under $P'$, in expectation, this number is less than $|P'|\cdot n/k$. Therefore, in expectation, they are not eligible to choose it.

 It is easy to see that ex post $\alpha$-core implies ex ante $\alpha$-core, since if for each $P$ in the support of a distribution that $\A_{k}$ returns and each $P'$, it holds that $\CV(P,P',\alpha)< |P'|\cdot n/k$, then  $\E_{P \sim \A_k}[\CV(P,P',\alpha)]<  |P'|\cdot n/k$.

\section{Fairness and Ex Post Core }\label{sec:fair-greedy-capture}

\begin{algorithm}[t] 
    \caption{$\fgc_{k}$}\label{alg:GC}
    \KwIn{ $[n]$,  $d$}
    \KwOut{  $P_{\ell}$ and $\lambda_{\ell}$, for $\ell \in [L]$,   where each $P_{\ell}$ represents a panel of size $k$ and $\lambda_{\ell}$ represents its probability of being selected}  
    \tcc{Create a  $(k/n)$-fractional allocation by distributing a $k/n$ fraction for each individual among $k$ balls, ensuring that each ball contains a total fractional amount equal to $1$.}
    $X\gets [0]^{k\times n}$;
    $\delta\gets 0$; $j\gets 1$; $\{y_i\gets k/n\}_{i \in [n]}$\;
    \While{$\sum_{i\in [n]} y_j>0$
    }{
            Smoothly increase $\delta$\; 
            \While{$\exists i\in [n]$,
            such that $ \sum_{i'\in B(i,\delta)} y_{i'}\geq 1$}{
                \While{$X_j=\sum_{i\in [n]} X_{j,i}<1$}{
                    Pick $i'\in B(i,\delta)$ with $x_{i'}>0$\;
                    $X_{j,i'}\gets \min(1-X_j, y_i)$\;
                    $y_i\gets y_i-X_{j, i'}$\;
                    }
                    $j\gets j+1$\;
                }
    }
    \tcc{Apply Birkhoff’s decomposition}
    $X' \gets [1/n]^{(n - k) \times n}$\;
    Let $Y = \left[
        \begin{array}{c}
        X \\
        X'
        \end{array}
    \right]$\;
    Compute a decomposition of $Y=\sum_{\ell=1}^L \lambda_{\ell} Y^{\ell}$ using the Birkhoff's decomposion (\Cref{theor:Birkhoff})\;
    \For{$\ell = 1$ to $L$}{$P_{\ell}\gets \big\{i \in [n]\mid Y^{\ell}_{j,i}=1 \text{ for some } j\leq k\big\}$}
    \Return{
    distribution over $L$ panels $\{P_\ell\}_{\ell \in [L]}$ where $P_\ell$ is selected with probability $\lambda_\ell$}
\end{algorithm}

In this section, we investigate if there are selection algorithms that are fair, and in addition,  provide a constant approximation to the ex post $q$-core. 
Unsurprisingly, uniform selection may fail to provide any bounded approximation to the ex post $q$-core for $q \in [k - 1]$~\footnote{ For $q = k$, we show in \Cref{app:US-ex-post-core} that all panels lie in the $2$ approximation of the $k$-core; hence, any algorithm including uniform selection provides an ex post $2$-$k$-core. }. This happens because each panel has a nonzero probability of selection, and there may exist panels with arbitrarily large violations of the $q$-core objective.

\begin{theorem}\label{thm:US-impos-ex-post-core}
For any $q \in [k-1]$ and  $\floor{\nicefrac{n}{k}}\ge k$,  there exists an instance such that  uniform selection is  not in the ex post $\alpha$-$q$-core for any bounded $\alpha$. 
\end{theorem}
\begin{proof}
    Consider an instance in which there are $\floor{\nicefrac{n}{k}}$ individuals  in group $A$ and the remaining individuals are in group $B$. Suppose that the distance between any two individuals in the same group is $0$, and the distance between any two individuals in different groups is $1$. Since, $\floor{\nicefrac{n}{k}}\geq k$, uniform selection has a non-zero probability of returning a panel where all the representatives are from group $A$. In this scenario, for any $q\in [k-1]$, the $q$-cost of all the individuals in group $B$ is equal to $1$. However, individuals in group $B$ are entitled to choose up to $k-1$ representatives among themselves, and if they do so, their $q$-cost becomes $0$, resulting in an unbounded improvement of their $q$-cost. Therefore, uniform selection is not in the ex post $\alpha$-$q$-core for any bounded $\alpha$. 
\end{proof}

Therefore, we ask: For every  $q$, is there any selection algorithm that keeps the fairness guarantee  of uniform selection  and ensures that every  panel in its support  is in the constant approximation of the $q$-core? We answer this positively.

We present  a selection algorithm, called $\fgc_{k}$, that is fair and  in the ex post $6$-$q$-core,  for {\em every} $q\in [k]$.  We highlight that the algorithm does not need to know the value of $q$. Our algorithm leverages  the basic idea of  the Greedy Capture algorithm  introduced by~\citet{chen2019proportionally}, which returns a panel in the {$(1 + \sqrt{2}) \approxeq 2.42$-approximation of the $1$-core. Note that this algorithm is deterministic and need not satisfy fairness. Briefly,  Greedy Capture starts with an empty panel and grows a ball around all individuals at the same rate. When a ball captures at least $\ceil{n/k}$ individuals for the first time, the center of the ball is included in the panel and all the captured individuals are disregarded. The algorithm keeps growing balls on all individuals, including the opened balls. As the opened balls continue to grow and capture more individuals, the newly captured ones are immediately disregarded as well. Note that the final panel can be of size less than $k$.

At a high level, $\fgc_{k}$, as outlined in \Cref{alg:GC}, operates as follows: it greedily opens $k$ balls using the basic idea of the Greedy Capture algorithm, ensuring each ball contains sufficiently many individuals. In contrast to Greedy Capture, which selects the centers of the balls as the representatives, our algorithm probabilistically selects precisely one individual from each of the $k$ balls.

Before, we describe the algorithm in more detail, we define a $(k/n)$-fractional allocation as a non-negative $k \times n$ matrix $X \in [0, 1]^{k \times n}$ where  entries in each row sums to $1$ and entries in each column sum to $k/n$, i.e., for each $i\in [n]$, $\sum_{j\in [k]}X_{j,i}=k/n$, and for each $j\in [k]$, $\sum_{i\in [n]}X_{j,i}=1$. 
The algorithm, during its execution, generates a \emph{$(k/n)$-fractional allocation} $X$ of individuals in $[n]$ into $k$ balls, where $X_{j,i}$ denotes the fraction of individual $i$ assigned to ball $j$. We say that an individual $i$ is assigned to ball $j$, if $X_{j,i}>0$. An individual can be assigned to more than one balls.

The $(k/n)$-fractional allocation $X$ is generated as follows. Denote the unallocated part of each individual $i$ by  $y_i$.  Start with $y_i = k/n$. This corresponds to the fairness criterion that we allocate a $k/n$ probability of selection to each individual.
\Cref{alg:GC}
grows a ball around every  individual in $[n]$ at the same rate. Suppose a ball captures individuals whose combined unallocated parts sum to at least $1$. Then, we open this ball  and from individuals $i'$ captured by this ball with $y_{i'} > 0$, we arbitrarily remove a total mass of exactly $1$ and assign it to the ball.  This can be done in various ways, e.g., greedily pick an individual  $i'$ with positive $y_{i'}$ and allocate $\min\{1 - \sum_{i \in [n]} X_{j, i}, y_{i'}\}$ fraction of it to the corresponding row (i.e. ball). This procedure terminates when the $k/n$ fraction of each individual is fully allocated.  Note that since each time a ball opens, a total mass of $1$ is deducted from $y_i$-s and, for each $i\in [n]$,   $y_i$ starts with a fraction of $k/n$,  exactly $k$ balls are opened.

\paragraph{Sampling panels from the $(k/n)$-fractional allocation.}  Next, we show a method of decomposing $X$, the $(k / n)$-fractional allocation, to a distribution over panels of size $k$ that each contain at least one representative from each ball.  We employ the  Birkhoff's decomposition \citep{birkhoff1946three}. This theorem applies over square matrices that are bistochastic. A matrix is bistochastic if every entry is nonnegative and the sum of elements in each of its rows and columns is equal to $1$. 

\begin{theorem}[Birkhoff-von Neumann]\label{theor:Birkhoff}
    Let $Y$ be a $n\times n$ bistochastic matrix. There exists a polynomial time algorithm that computes a decomposition  $Y = \sum_{\ell = 1}^L \lambda_\ell Y^{\ell}$, with $L\leq n^2-n+2$, such that for each $\ell \in [L]$, $\lambda_\ell \in [0,1]$, $Y^{\ell}$  is a permutation matrix and $ \sum_{\ell = 1}^L \lambda_\ell =1$.
\end{theorem}

We cannot directly apply the theorem above, since the $(k/n)$-fractional allocation $X$ is not bistochastic nor a square matrix. However, we can complete $X$ into a square matrix $Y = \left[\begin{array}{c}
     X \\
     X' 
\end{array}\right]$
by adding $n - k$ rows $X' = [1/n]^{(n - k) \times n}$ where all entries are $1/n$. Note that the resulting matrix $Y$ is bistochastic. Indeed, each row of both $X$ and $X'$ sums to $1$ by their definition; further, as each column of $X$ sums to $k/n$ and that it is followed by $n - k$ of $1/n$ entries in $X'$, the columns also sum to $1$. Note that there are various choices of $X'$ that makes $Y$ a bistochastic matrix, but here we use the uniform matrix for simplicity.
Then, the algorithm applies \Cref{theor:Birkhoff} and computes the decomposition $Y = \sum_{\ell = 1}^L \lambda_\ell Y^{\ell}$. For each permutation matrix $Y^{\ell}$, we create a panel $P_{\ell}$ consisting of the individuals that have been assigned to the first $k$ rows, i.e. $P_{\ell}$ contains all $i$-s with $Y^{\ell}_{j,i}=1$ for some $j\leq k$.  Finally, the algorithm returns the distribution that selects each panel $P_\ell$ with probability equal to $\lambda_\ell$.  

To prove that $\fgc_{k}$ is fair and   ex post $O(1)$-$q$-core, we need the next two  lemmas.

\begin{lemma} \label{lem:partition}
    Let $S\subseteq [n]$,  $P'$ be a panel, and $m = \floor{|P'| / q}$. 
    \begin{enumerate}
        \item 
    There exists a partitioning of $S$ into $m$ disjoint sets $T_1, \ldots, T_m$  and an individual $i^*_{\ell} \in T_\ell$ such that for all $\ell \in [m]$ and $i \in T_{\ell}$, $c_q(i,P') \leq c_q(i^*_{\ell},P')$ and $\top_q(i,P') \cap \top_q(i^*_{\ell},P') \neq \emptyset$.
        \item 
    There exists a partitioning of $S$ into $m$ disjoint sets $T_1, \ldots, T_m$  and an individual $i^*_{\ell} \in T_\ell$ such that for all $\ell \in [m]$ and $i \in T_{\ell}$, $c_q(i,P') \geq c_q(i^*_{\ell},P')$  and $\top_q(i,P') \cap \top_q(i^*_{\ell},P') \neq \emptyset$.
    \end{enumerate}
\end{lemma}
\begin{proof}

We start by showing the first part. 
We partition all the individuals in $S$ into $m \le \floor{\nicefrac{|P'|}{q}}$ groups, denoted by $T_1,\ldots, T_{m}$ iteratively as follows.

Suppose $i^*_1$ is the individual with the smallest $q$-cost over $P'$ (ties are broken arbitrary), i.e. $i^*_1 = \argmax_{i\in S} c_q(i,P')$. Then, $T_1$ is the set of all the individuals whose $q$ closest representatives from $P'$ includes at least one member of $\top_q(i^*_1,P')$, i.e.
\begin{align*}
 T_1=\{i \in S:  \top_q(i,P') \cap \top_q(i^*_1,P') \neq \emptyset \}.
\end{align*} 
Next, from the remaining individuals, suppose $i^*_2$ is the one with the smallest $q$-cost over $P'$, i.e.  $i^*_2 = \argmin_{i \in S \setminus T_1} c_q(i, P)$. Construct $T_2$ from $S \setminus T_1$ similarly by taking all the individuals whose at least one of their $q$ closest representatives in $P'$ is included in $\top_q(i^*_2,P')$. We repeat this procedure, and in round $\ell$, we find $i^*_\ell \in S \setminus (\cup_{\ell'=1}^{\ell-1} T_{\ell'})$ that has the smallest cost over $P'$,  and construct $T_{\ell}$ by assigning any individual in $S \setminus (\cup_{\ell'=1}^{\ell-1} T_{\ell'})$ whose at least one of the $q$ closest representatives belongs in $\top_q(i^*_{\ell},P')$. Note that for any $\ell_1, \ell_2 \in [m]$ with $\ell_1<\ell_2$, $\top_q(i^*_{\ell_1},P') \cap \top_q(i^*_{\ell_2},P')=\emptyset $, as if  at least one of the  $q$ closest representatives of $i^*_{\ell_2}$ in $P$ is included in $\top_q(i^*_{\ell_1},P')$,  then  $i^*_{\ell_2}$ would have been assigned to $T_{\ell_1}$ and would not belong in $S \setminus (\cup_{\ell'=1}^{\ell_2-1} T_{\ell'})$. This means that in each round, we consider $q$ representatives that have not been considered before, and hence after $\floor{|P'|/q}$ rounds, less than $q$ representatives in $P'$ may remain unconsidered. As a result, after at most $\floor{|P'|/q}$ rounds,  all the   individuals will have  been assigned to some group, since at least one of their $q$ closest representatives has been considered.  

The second part follows by simply setting $i^*_{\ell}$ to be equal to the individual in  $ S \setminus (\cup_{\ell'=1}^{\ell-1} T_{\ell'})$ that has the largest cost over $P'$, i.e. $i^*_{\ell} = \argmin_{i \in  S \setminus (\cup_{\ell'=1}^{\ell-1} T_{\ell'})} c_q(i, P)$. All the remaining arguments remain the same. 
\end{proof}

\begin{lemma}\label{lem:cost-distance}
   For any panel $P$ and any $i,i'\in [n]$, it holds that  $
    c_q(i,P)\leq d(i,i')+c_q(i',P)   
    $.
\end{lemma}
\begin{proof}
    Consider a ball centered at $i'$ with radius $c_q(i',P)$. This ball contains at least $q$ representatives of $P$. Hence, $ c_q(i,P)$ is less than or equal to the distance of $i$ to one of the $q$ representatives that are included in $B(i', c_q(i',P))$ which is at most $d(i,i')+c_q(i',P)$. 
\end{proof}

Now, we are ready to prove the next theorem.

\begin{theorem}\label{thm:ex-post-core}
For every $q$, $\fgc_{k}$ is fair and  in the ex post $6$-$q$-core. 
\end{theorem}
\begin{proof}
Seeing that the algorithm is fair is straightforward. For a matrix $A$, let $A[1:k,:]$ be the submatrix induced by keeping its first $k$ rows.  First, note that  for each panel $P^\ell$ we choose the individuals that have been assigned to   $Y^{\ell}[1:k,:]$ and  second, recall that  $Y[1:k,:]=X$. The fairness of the algorithm follows by the facts  that   $Y[1:k,:]=X=\sum_{\ell=1}^L \lambda_{\ell} Y^{\ell}[1:k,:]$ and  for each $i\in [n]$, $\sum_{j=1}^k X_{j,i}=k/n$.

 We proceed by showing that $\fgc_{k}$ is in the ex post $6$-$q$-core, for all $q \in [k]$.
First, note that if an individual $i$ is assigned to a ball $j$ in some $Y^{\ell}$, then we must have $X_{j, i} > 0$. Now, 
since each individual $i \in [n]$ is assigned to a ball $j \in [k]$ in the permutation, we get that  that  at least one individual is selected from each ball.

 Let $P$ be any panel that the  algorithm may return.  
 Suppose for contradiction that there exists a panel $P'$  such that $\CV(P,P',6)\geq |P'| \cdot n/k $. This means that there exists $S\subseteq[n]$, with $|S|\geq |P'|\cdot n/k$, such that: 
\begin{align}\label{eq:fgc-ex-post}
  \forall i \in S,  \quad \quad c_{q}(i,P)> 6 \cdot  c_{q}(i,P').
\end{align} 

Let $T_1,\ldots, T_m$ be a partition of $S$ with respect to $P'$, as given in the first part of~\Cref{lem:partition}. Since $m\leq \floor{|P'|/q}$ and $|S|\geq |P'|\cdot n/k$, we conclude that there exists a part, say $T_{\ell}$, that has size at least $ q \cdot n/k$. From~\Cref{lem:partition}, we know  that there exists $i^*_{\ell}\in T_{\ell}$ such that for each $i\in T_{\ell}$ it holds that   $c_q(i,P') \leq c_q(i^*_{\ell},P')$ and $\top_q(i,P') \cap \top_q(i^*_{\ell},P')\neq \emptyset$. Therefore, we can conclude that for each $i\in T_{\ell}$, $d(i^*_{\ell},i)\leq 2\cdot c_q(i^*_{\ell},P')$, as following: Pick an arbitrary  representative in $\top_q(i,P') \cap \top_q(i^*_{\ell},P')$ and denote it as $r_i$. Then,
\begin{align*}
    d(i, i^*_{\ell})\leq d(i,r_i)+ d(r_i, i^*_{\ell})\leq  c_{q}(i,P') +   c_{q}(i^*_{\ell},P')\leq  2\cdot c_{q}(i^*_{\ell},P'). 
\end{align*}
This implies that the ball centered at $i^*_{\ell}$ with a radius of $2\cdot c_{q}(i^*_{\ell},P')$ captures  all individuals in $T_{\ell}$.

Now, consider all the balls that $\fgc_{k}$ opens and contain individuals from $T_{\ell}$. Since $T_{\ell}\geq q\cdot n/k$ and each ball is assigned a total fraction of $1$, there are at least $q$ such balls. Next, we claim that least  $q$ of them have radius at most $2\cdot c_{q}(i^*_{\ell},P')$. Suppose for contradiction that  at most $q-1$ of them  have 
  radius at most $2 \cdot c_q(i^*_\ell, P')$.  This means that a total fraction of at least $1$ from individual in $T_{\ell}$ is assigned to balls with radius strictly larger than $2 \cdot c_q(i^*_\ell, P')$.  However, the ball centered at $i^*_{\ell}$ with radius  $2 \cdot c_q(i^*_\ell, P')$ would have captured this fraction, and therefore we reach a contradiction.

Next, denote with $B_1, \ldots, B_q$,  $q$ balls that are opened,  and each contain individuals from $T_{\ell}$ and have radius at most $2\cdot c_{q}(i^*_{\ell},P')$. Due to the definition of $\fgc_k$,  each panel  that is returned,  contains at least one representative from each ball. Therefore, each  ball $B_j$ contains at least one representative, denoted by $r_j$. Now, note that since each $B_j$ contains at least one individual from $T_{\ell}$, denoted by $i_j$, we have that
\begin{align*}
\forall j \in [q], \quad    d(i^*_{\ell}, r_j)\leq d(i^*_{\ell}, i_j)+ d(i_j,r_j)\leq 6\cdot c_{q}(i^*_{\ell},P'),
\end{align*}
where the first inequality follows from the triangle inequality and the last inequality follows from the facts that  for each $i\in T_{\ell}$,  $d(i^*_{\ell}, i)\leq 2\cdot c_{q}(i^*_{\ell},P')$, and  each $B_j$ has radius at most $2\cdot c_{q}(i^*_{\ell},P')$ and both $i_j$ and $r_j$ belong to this ball. Therefore, there are at least $q$ representatives in $P$ that have distance at  most $6\cdot c_{q}(i^*_{\ell},P')$ from $i^*_\ell$. But then,  $ c_{q}(i^*_{\ell},P)\leq 6\cdot c_{q}(i^*_{\ell},P')$ which is a contradiction with~\Cref{eq:fgc-ex-post}.  
\end{proof}

As we discussed above,  the ex post $\alpha$-$q$-core  implies the ex ante $\alpha$-$q$-core which means that $\fgc_{k}$ is also in the ex ante $6$-$q$-core for all $q\in [k]$. In the next section, we show that no fair algorithm provides an approximation better than $2$ to the ex ante $q$-cost, for any $q$. Therefore, we get that no fair selection algorithm provides an approximation better than $2$ to the ex post $q$-core either. This means that $\fgc_{k}$ is  optimal up to a  factor of $3$. 

\subsection{Ex Post Core and Quotas over Features}

In our introduction, we discussed a common approach used to ensure proportional representation, which involves setting quotas based on individual or groups of features. For instance, a quota might mandate that at least $45\%$ of representatives are female. While the concept of the core aims to achieve proportional representation across intersecting features, it may not guarantee the same across individual features. For instance, a panel comprising entirely men could still meet core criteria, even if the overall population is 50\% women. This raises the question of whether it's possible to achieve both types of representation to the degree that is possible. We argue that this is feasible and show how the core requirement can be translated into a set of quotas.

As showed above, $\fgc_k$ generates $k$ balls, with each individual assigned to one or more balls. The key condition for  achieving an ex post $O(1)$-$q$-core is to have at least one representative from each ball.  This condition can be transformed into quotas by introducing an additional feature, $b_i$, for each individual $i$, indicating the balls they belong to. Thus, $b_i$ can take values in $2^{[k]}\setminus \{\emptyset\}$, where $2^{[k]}$ represents the power set of $[k]$. We then can set quotas that require the panel to contain at least one representative $i$ that belongs in ball $j$, i.e. $j\subseteq b_i$,  for each $j\in [k]$.  In other words, we can think of each ball as a subpopulation from which we want to draw a representative. We can then utilize the methods proposed by~\citet{FGGH+21} to identify panels that meet these quotas, along with others as much as possible, while maximizing fairness. We also note that this translation allows for sampling from a biased pool of representatives using the algorithm of the aforementioned paper, as long as the characteristics of the global population are known and the balls can be constructed based on them.

\section{Uniform Selection and Ex Ante Core}\label{sec:uniform-selection}

We have already discussed that uniform selection fails to provide any reasonable approximation to the ex post $q$-core, for almost all values of $q$. However, as we mentioned in the introduction, it seems to satisfy the ex ante $q$-core, at least when $k$ is very large. In this section, we ask whether indeed uniform selection  satisfies a constant approximation of the ex ante $q$-core, in a rigorous way, for all values of $q$ and $k$. We show that  uniform selection is in the ex ante $4$-$q$-core, for every $q$.~\footnote{
In fact, for $q = k$, uniform selection is in the ex ante $k$-core (see \Cref{app:US-ex-ante-k}). 
The main reason is that, for $q = k$, it suffices to show that the grand coalition does not deviate ex-ante. Since each panel is selected with non-zero probability, the marginal probabilities of deviation is strictly less than one, and the ex ante $k$-core is satisfied.
 }

To show  this result, we use the following form of Chu–Vandermonde identity which we prove in~\Cref{append:CV-identity} for completeness.

\begin{lemma}[Chu–Vandermonde identity]\label{lem:chu-vandermonde}
For any  $n, k,$ and $r$,  with $0 \leq r \leq k \leq n$, it holds
\begin{align*}
\sum_{j=0}^n \binom{j}{r} \cdot \binom{n-j}{k-r} =\binom{n+1}{k+1}.
\end{align*} 
\end{lemma}

Now, we are ready to prove the following theorem.
\begin{theorem}\label{theor-exp-core-sort}
For any $q$,   uniform selection is in the ex ante $4$-$q$-core, i.e. for any panel $P'$
\begin{align*}
\E_{P\sim \sort}\left[\CV(P,P',4)\right]< |P'| \cdot \frac{n}{k}.
\end{align*}
\end{theorem}
\begin{proof}
Let  $P'$ be any panel. By linearity of expectation, we have that
\begin{align*}
    \E_{P\sim \sort}\left[\CV(P,P',4)\right]= \sum_{i\in [n]} \; \Pr\nolimits_{P\sim \sort}\left[c_q(i,P)>4 \cdot c_q(i,P')\right]. 
\end{align*}
Let $T_1\ldots, T_m$ be a partition of $[n]$ with respect to $P'$, as given in the second part of~\Cref{lem:partition}. 
For each $\ell\in [m]$, we reorder the individuals in  $T_{\ell}$ in an increasing order  based on their distance from $i^*_{\ell}$, and  relabel them as $i^{\ell}_1,\ldots, i^{\ell}_{|T_{\ell}|}$. This way,  $i^{\ell}_1$ and $i^{\ell}_{|T_{\ell}|}$ are the individuals in $T_{\ell}$ that have the smallest and the largest distance from $i^*_{\ell}$, respectively. Then, we get that
\begin{equation}   
     \sum_{i\in [n]} \; \Pr\nolimits_{P\sim \sort} \left[c_q(i,P)>4 \cdot c_q(i,P')\right]
     =
     {\sum_{\ell=1}^m \sum_{j=1}^{|T_{\ell}|}}
     \; \Pr\nolimits_{P\sim \sort}
     \left[c_q(i_j^\ell,P)>4 \cdot c_q(i_j^{\ell},P')\right].\label{ineq:4-exp-appr-2} 
\end{equation}

In the next lemma, we  bound $\Pr_{P\sim \sort} \left[c_q(i_j^\ell,P)>4 \cdot c_q(i^\ell_j,P')\right]$ for each $i_j^\ell$.  

\begin{lemma}\label{lem:US-ex-ante}
   For each $\ell\in [m]$ and  $j\in [|T_{\ell}|]$,
   \[
   \Pr\nolimits_{P\sim \sort} \left[c_q(i_j^\ell,P)>4 \cdot c_q(i^\ell_j,P')\right]\leq
   {\sum_{r=0}^{q-1} \frac{1}{{\binom{n}{k}}}} \cdot \binom{j}{r} \binom{n-j}{k-r}.
   \]
\end{lemma}
\begin{proof}
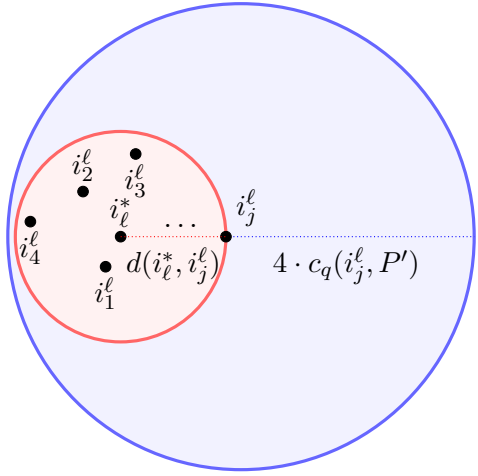
\begin{figure}[t]
\centering
\begin{tikzpicture}
 \filldraw[ color=blue!60, fill=blue!5, very thick](1.6,0) circle (3.1);

\draw[ densely dotted,blue] (1.4,0) -- (4.7,0);
\filldraw[black] (3,0) circle (0pt) node[anchor=north]{$4\cdot c_q(i^{\ell}_j,P')$};
\filldraw[color=red!60, fill=red!5, very thick](0,0) circle (1.4);
\filldraw[black] (0,0) circle (2pt) node[anchor=south]{$i^*_{\ell}$};
\filldraw[black] (1.4,0) circle (2pt) node[anchor=south west]{$i^{\ell}_j$};
\filldraw[black] (0.7,0) circle (0pt) node[anchor=north]{$d(i^*_{\ell},i^{\ell}_j)$};
\filldraw[black] (-0.2,-0.4) circle (2pt) node[anchor=north]{$i^{\ell}_1$};
\filldraw[black] (-0.5,0.6) circle (2pt) node[anchor=south]{$i^{\ell}_2$};
\filldraw[black] (0.2,1.1) circle (2pt) node[anchor=north]{$i^{\ell}_3$};
\filldraw[black] (-1.2,0.2) circle (2pt) node[anchor=north]{$i^{\ell}_4$};
\filldraw[black] (0.8,0.3) circle (0pt) node[anchor=north]{$\ldots$};
\draw[ densely dotted,red] (0,0) -- (1.4,0);

\end{tikzpicture}
\caption{Diagram for Proof of~\Cref{lem:US-ex-ante}  } \label{fig:exp-core-sort}
\end{figure}
For each $i_j^\ell$, let $r^{\ell}_j$ be  an arbitrary representative   in  $\top_{q}(i^\ell_j,P') \cap \top_{q}(i^*_\ell,P')$. Then, we get  that
\begin{equation} 
   d(i_j^\ell,i^*_{\ell})  \leq d(i_j^\ell,r^{\ell}_j ) + d(r^{\ell}_j, i^*_\ell )
    \leq  c_q(i^\ell_j,P') +  c_q(i^*_\ell,P')\leq  2 \cdot c_q(i^\ell_j,P'), \label{ineq:4-exp-appr-1} 
\end{equation}
where  the last inequality follows from the fact that  $i^*_{\ell}$ has the smallest cost over $P'$ among all the individuals in $T_{\ell}$. Now, consider the ball that is centered at $i_j^\ell$ and has radius $4 \cdot c_q(i^\ell_j,P')$. Note that this ball contains any  individual $i^\ell_{j'}$ with $j'<j$. Indeed, for each $i^\ell_j$ and $i^\ell_{j'}$ with $j'<j$, we have that
\begin{align*}
    d(i^\ell_j, i^\ell_{j'})\leq  d(i^\ell_j, i^*_\ell) + d( i^*_{\ell}, i^\ell_{j'}) \leq 2\cdot  d(i^\ell_j, i^*_\ell)\leq  4 \cdot c_q(i^\ell_j,P'),
\end{align*}
where   the second inequality follows form the fact that for each $j',j \in [|T_{\ell}|]$ with $j'<j$,  $d(i^\ell_{j'},i^*_\ell) \leq d(i^\ell_{j},i^*_\ell)$ and the last inequality follows form \Cref{ineq:4-exp-appr-1}.  This argument is drawn  in  \Cref{fig:exp-core-sort}. 

When $c_q(i_j^\ell,P)>4 \cdot c_q(i^\ell_j,P')$, then we get that $|P\cap \{ i^{\ell}_1,\ldots, i^{\ell}_{j} \}|<q$, as otherwise there would exist at least $q$ individuals in 
$B(i^\ell_j,4 \cdot c_q(i^\ell_j,P'))$, and $c_q(i_j^\ell,P)$ would be at most $4 \cdot c_q(i^\ell_j,P')$. Hence, we have that 
\begin{align*}
     \Pr_{P\sim \sort} [c_q(i_j^\ell,P)>4 \cdot c_q(i^\ell_j,P')] &\leq  \Pr_{P\sim \sort} \left[~\lvert P\cap \{ i^{\ell}_1,\ldots, i^{\ell}_{j} \}\rvert <q\right] \\
     &=\Pr_{P\sim \sort} \left[
     \bigcup_{r=0}^{q-1} \, \lvert  P\cap \{ i^{\ell}_1,\ldots, i^{\ell}_{j} \}\rvert = r \right] \\
    & \leq \sum_{r=0}^{q-1} \Pr_{P\sim \sort} \left[\,\lvert P\cap \{ i^{\ell}_1,\ldots, i^{\ell}_{j} \}\rvert = r\right]
    =
   {\sum_{r=0}^{q-1} \frac{1}{{\binom{n}{k}}}} \cdot \binom{j}{r} \binom{n-j}{k-r}.
\end{align*}
where the second inequality follows from the  Union Bound and the last equality follows form the fact that uniform selection chooses $k$ out of $n$ individuals uniformly at random.  
\end{proof}

Then, by returning to \Cref{ineq:4-exp-appr-2}
we get that, 
\begin{align*}
    \E_{P\sim \sort}[\CV(P,P',4)]&= \sum_{\ell=1}^m \; \sum_{j=1}^{|T_{\ell}|} \; \Pr\nolimits_{P\sim \sort}\left[c_q(i_j^\ell,P)>4 \cdot c_q(i_j,P')\right]
    \\
   &
   \leq 
   \frac{1}{{\binom{n}{k}}} \cdot 
   \sum_{\ell=1}^m \; \sum_{j=1}^{|T_{\ell}|}  \; \sum_{r=0}^{q-1} \; \binom{j}{r} \binom{n-j}{k-r}
   & \text{(by \Cref{lem:US-ex-ante})}
   \\
   &
   =
   \frac{1}{{\binom{n}{k}}} \cdot 
   \sum_{\ell=1}^m \; \sum_{r=0}^{q-1} \; \sum_{j=1}^{|T_{\ell}|}  \; \binom{j}{r} \binom{n-j}{k-r}
   & \text{(swap summations)}
   \\
   &
   \leq 
   \frac{1}{{\binom{n}{k}}} \cdot 
   \sum_{\ell=1}^m \; \sum_{r=0}^{q-1} \; \sum_{j=0}^{n} \binom{j}{r} \binom{n-j}{k-r}
   & (|T_{\ell}| \le n)
   \\
   &
   =
   \frac{1}{{\binom{n}{k}}} \cdot 
   \sum_{\ell=1}^m \; \sum_{r=0}^{q-1}  \; \binom{n+1}{k+1}
   &
   (\text{by \Cref{lem:chu-vandermonde}})
   \\
   &=
   \sum_{\ell=1}^m \; \sum_{r=0}^{q-1}  \frac{n+1}{k+1}= m\cdot q \cdot \frac{n+1}{k+1} < |P'| \cdot \frac{n}{k},
\end{align*}
where the last inequality follows from the facts that $m \le \floor{ \frac{|P'|}{q}}$ and ${\frac{n+1}{k+1} < \frac{n}{k}}$ for $k<n$. 
\end{proof}

In the next theorem, we show that for any  $q<k$, no selection algorithm that is fair, is guaranteed  to achieve ex ante $\alpha$-$q$-core with $\alpha<2$, and hence uniform selection is  optimal up to a factor of $2$. 
\begin{theorem}\label{thm:lower-bound}
    For any $q\in [k-1]$, when $n\geq 2k^2/(k-q)$, there exists an instance such that no selection algorithm that is fair,  is in the ex ante  $\alpha$-$q$-core with $\alpha<2$.
\end{theorem}
\begin{proof}
    Consider a star graph with $n - q$ leaves and an internal node. Suppose $q$ individuals $I = \{i_1,\ldots, i_q\}$ lie on the internal node, and exactly one individual lies on each of the $n - q$ leaves. Individuals in $I$ have a distance of $0$ from each other and a distance of $1$ from $[n] \setminus I$; and, the distance between a pair of individuals from $[n] \setminus I$ is equal to $2$. These distances satisfy the triangle inequality. 
    
    Let $P$ be an arbitrary panel of size $k$ that does not contain $i_1$. We show that for $P'=I$ and every $\alpha<2$, we have that
    $
        \CV(P,P',\alpha)\ge n-k.
    $
    For any $i \in I$, it holds $c_q(i,P)=1$ and 
    $c_q(i,P')=0$ --- which is an unbounded improvement. For any individual $i$ in $[n] \setminus (I \cup P)$,  $c_q(i,P)=2$ since their $q$th closest representative in $P$ would be on another leaf, while $c_q(i,P')=1$ --- which yields a $2$ factor improvement. Therefore, $\CV(P,P',\alpha)\ge |([n] \setminus (I \cup P)) \cup I| \ge n - |P| =  n-k$, for every $\alpha<2$.
    
    Under any fair selection algorithm, $i_1$ is not included in the panel with probability $1-\nicefrac{k}{n}$. Thus, we have that
    \begin{align*}
         \E_{P \sim \sort} [\CV(P,P',\alpha)] \ge \Pr_{P \sim \sort} [i_1 \notin P] \cdot (n - k) = (1-k/n) \cdot (n-k) \geq q\cdot n/k= |P'| \cdot n/k,
    \end{align*}
    where the last inequality follows from the assumption that $n\geq 2k^2/(k-q)$.  
\end{proof}

\section{Auditing Ex Post Core}
\label{sec:audit}

\begin{algorithm}[t]
\caption{Auditing Algorithm}\label{alg:auditing}
\KwIn{$P$,  $[n]$,  $d$, $k$, $q$, }
\KwOut{$\hat{\alpha}$}
\For{$j \in [n]$}{
    $\hP_{j}\gets$ $\{j\} \cup q-1$ closest neighbors of $j$;\\
    $\hat{\alpha}_j\gets$ the $\ceil{q\cdot n/k}$ largest value among $\{c_q(i,P)/ c_q(i,\hP_j)\}_{i\in [n]}$;\label{audit_alg_line:ratio}\\
}  
\Return{$\hat{\alpha}\gets\argmax_{j\in [n]}\hat{\alpha}_j$}
\end{algorithm}
In this section, we turn our attention to the following question: Given a  panel $P$, how much does it violate the $q$-core, i.e. what is the maximum value of $\alpha$ such that there exists a panel $P'$ with $\CV(P,P',\alpha)\geq |P'| \cdot n/k$? This auditing question can be very useful in practice for measuring the proportional representation of a panel formed using a method that does not guarantee any panel to be in the approximate core, such as uniform selection.

\citet{chen2019proportionally} ask the same question for the case where the cost of an individual for a panel is equal to her distance form her closest representative in the panel, i.e. when $q=1$.  In this case, it suffices to restrict our attention to panels of  size $1$, which are subsets of the population that individuals may prefer to be represented by. In other words, given a panel $P$, we can simply consider every individual as a potential representative and check if a sufficiently large subset of the population prefers this individual to be their representative over $P$. Thus, we can find the maximum  $\alpha$ such that there exists $P'$, with $\CV(P,P',\alpha)\geq n/k$ as following: For each $j\in [n]$,  calculate $\alpha_j$ which is equal  to  the $\ceil{n/k}$ largest value among the set $\{c_q(i,P)/c_q(i,\{j\})\}_{i \in [n]}$ containing  the $q$-cost ratios of $P$ to $\hP$.  Then, $\alpha$ is equal to the maximum value among all $\alpha_j$'s.

For $q>1$, this question is more challenging. We show the possibility of approximating the value of the maximum $\alpha$, by generalizing the above procedure as following: For each $j\in [n]$, let $\hP_j$ be the panel that contains $j$ and its  $q-1$ closest neighbors. Then, calculate $\hat{\alpha}_j$ as  the $\ceil{q \cdot n/k}$ largest value of among the set $\{c_q(i,P)/c_q(i,\hP_j)\}_{i \in [n]}$.  Then, we return  the maximum value among all $\hat{\alpha}_j$'s as $\hat{\alpha}$.
\Cref{alg:auditing} executes this procedure.  We show that the maximum  $\alpha$  such that there exists a panel $P'$ with $\CV(P,P',\alpha)\geq |P'| \cdot n/k$
is at most $3\cdot \hat{\alpha}+2$. 

\begin{restatable}{theorem}{auditpanel}
\label{thm:auditing}
There exists an efficient algorithm that for every panel $P$ and $q \in [k]$ returns $\hat{\alpha}$-$q$-core violation that satisfies $\hat{\alpha} \le \alpha \le 3\hat{\alpha} + 2$, where $\alpha$ is the maximum amount of $q$-core violation of $P$.
\end{restatable}

\begin{proof}
    Suppose for contradiction that while the algorithm returns $\hat{\alpha}$,   there exists  $S\subseteq [n]$ and $P'\subseteq [n]$, with $|S|\geq |P'| \cdot n/k$,  such that 
          \begin{align*}
              \forall i \in S,  \quad \quad c_{q}(i,P)> (3 \cdot \hat{\alpha}+2) \cdot  c_{q}(i,P').
          \end{align*} 
   
    First, note that  if the algorithm outputs $\hat{\alpha}$, this means that for every  $\alpha'>\hat{\alpha}$ and $j\in [n]$, it holds that 
    \begin{equation} \label{ineq:size-PS}
        \CV(P,\hP_{j},\alpha')<|\hP_{j}|\cdot n/k, 
    \end{equation}
    as otherwise the algorithm would output a value strictly larger than $\hat{\alpha}$.  

  Let $T_1,\ldots, T_m$ be a partition of $S$ with respect to $P'$, as given in the first part of~\Cref{lem:partition}. Since $m\leq \floor{|P'|/q}$ and $|S|\geq |P'|\cdot n/k$, we conclude that there exists a part, say $T_{\ell}$, that has size at least $ q \cdot n/k$.  Moreover, since there exists $i^*_{\ell}\in T_{\ell}$ such that for each   $i\in T_{\ell}$, it holds that $c_q(i,P') \leq c_q(i^*_{\ell},P')$ and $\top_q(i,P') \cap \top_q(i^*_{\ell},P')\neq \emptyset $, we can conclude that $d(i,i^*_{\ell})\leq 2\cdot c_q(i^*_{\ell}, P')$, by considering a representative in $\top_q(i,P') \cap \top_q(i^*_{\ell},P')$ and applying the triangle inequality,  i.e. $d(i,i^*_{\ell})\leq d(i,r_i)+d(r_i,i^*_{\ell})\leq 2\cdot c_q(i^*_{\ell}, P')$, where $r_i\in \top_q(i,P') \cap \top_q(i^*_{\ell},P')$. 
    This means there exists a ball centered at $i^*_{\ell}$ that has radius $2\cdot c_{q}(i^*_{\ell},P')$ and captures all the individuals in $T_{\ell}$.      Now, note that there exists $i'\in T_{\ell}$ such that
    $\alpha' \cdot c_q(i',\hP_{i^*_{\ell}})> c_q(i',P)$, since otherwise for each $i\in T_{\ell}$ would hold that  $\alpha' \cdot c_q(i',\hP_{i^*_{\ell}})\leq c_q(i',P)$ and then $\CV(P,\hP_{i^*_{\ell}}, \alpha')\geq q\cdot n/k=|\hP_{i^*_{\ell}}|\cdot n/k $ which contradicts \Cref{ineq:size-PS}. Hence, 
    \begin{align*}
    c_{q}(i^*_{\ell},P)
    \leq d(i^*_{\ell},i') +   c_q(i', P) 
  < d(i^*_{\ell},i') + \alpha'\cdot  c_q(i',\hP_{i^*_{\ell}})  
    &\leq d(i^*_{\ell},i')  + \alpha' \cdot (d(i^*_{\ell},i')+   c_q(i^*_{\ell},\hP_{i^*_{\ell}})\\
    &\leq   (3\cdot \alpha'+2)\cdot c_q(i^*_{\ell}, P').
    \end{align*}
where the first and the third inequalities follows from~\Cref{lem:cost-distance} and the last inequality follows from the facts that for each $i\in T_{\ell}$, $d(i,i^*_{\ell})\leq 2\cdot  c_{q}(i^*_{\ell},P')$, and  $ c_q(i^*_{\ell},\hP_{i^*_{\ell}})  \leq  c_q(i^*_{\ell}, P') $ for each $P'$ since $\hP_{i^*_{\ell}}$ consists of the $q$ closest neighbors of $i^*_{\ell}$. Therefore, $c_{q}(i^*_{\ell},P)  \leq   (3\cdot \hat{\alpha}+2)\cdot c_q(i^*_{\ell}, P')$ and the theorem follows.
\end{proof}

\section{Experiments}\label{sec:experiments}
In previous sections, we examined uniform selection from a worst-case perspective and found that it cannot guarantee panels in the core for any bounded approximation ratio. But, what about the average case? How much better is $\fgc$ than uniform selection in terms of their approximations to the ex post core in the average case? In this section, we aim to address these questions through empirical evaluations of both algorithms using real databases.

\subsection{Datasets} In accordance with the methodology proposed by \citet{SKMPS22}, we utilize the same two datasets used by the authors as a proxy for constructing the underlying metric space. These datasets capture various characteristics of populations across multiple observable features. It is reasonable to assume that individuals feel closer to others who share similar characteristics. Therefore, we construct a random metric space using these datasets.

\paragraph{\adult{}.} The first is the \adult{} dataset, extracted from the 1994 Current Population Survey by the US Census Bureau and available on the UCI Machine Learning Repository under a CC BY 4.0 license \cite{kohavi1996adult,DG17}. Our analysis focuses on five demographic features: \texttt{sex},
\texttt{race}, \texttt{workclass}, \texttt{marital.status}, and \texttt{education.num}. The dataset is comprised of $32{,}561$ data points, each with a sample weight attribute (\texttt{fnlwgt}). We identify $1513$ unique data points by these features and treat the sum of the weights associated with each unique point as a distribution across them.

\paragraph{\ess{}.} The second dataset we analyze is the European Social Survey (ESS), available under a CC BY 4.0 license \cite{ESS-9}. Conducted biennially in Europe since 2001, the survey covers attitudes towards politics and society, social values, and well-being. We used the ESS Round 9 (2018) dataset, which has $46{,}276$ data points and $1451$ features across 28 countries. On average, each country has around $250$ features (after removing non-demographic and country-unrelated data), with country-specific data points ranging from $781$ to $2745$. Each ESS data point has a post-stratification weight (\texttt{pspwght}), which we use to represent the distribution of the data points. Our analysis focuses on the ESS data for the United Kingdom (\ess{}-UK), which includes 2204 data points.

\subsection{Representation Metric Construction}
In line with the work of \citet{SKMPS22}, we apply the same approach to generate synthetic metric preferences, which are used to measure the dissimilarity between individuals based on their feature values.
Our datasets consist of two types of features: \emph{categorical} features  (e.g. sex, race, and martial status) and \emph{continuous} features  (e.g. income).
We define the distance between individuals $i$ and $j$ with respect to feature $f$ as follows:
\[
d(i, j; f) \coloneqq
\begin{cases}
    \mathbbm{1}[f(i) \ne f(j)], & \text{$\quad$ if $f$ is a \emph{categorical} feature;}\\
    {\frac{1}{{\max_{i', j'} |f(i') - f(j')|}}} \cdot |f(i) - f(j)|, & \text{$\quad$ if $f$ is a \emph{continuous} feature,}
\end{cases}
\]
where the normalization factor for continuous features ensures that $d(i, j; f) \in [0, 1]$ for all $i$, $j$, and $f$, and that the distances in different features are comparable. Next, we define the distance between two individuals as the weighted sum of the distances over different features, i.e.
$
d(i, j) = \displaystyle\sum\nolimits_{{f \in F}} \, w_f \cdot d(i, j; f),
$
where the weights $w_f$'s are randomly generated. Each unique set of randomly generated feature weights results in a new representation metric.

We generate $100$ sets of randomly-assigned feature weights per dataset, calculate a representation metric for each set, and report the performance metrics averaged over  $100$ instances.
Given that our datasets are samples of a large population (i.e. millions) and represented through a relatively small number of unique data points (i.e. few thousands), we assume that each data point represents a group of at least $k$ people, which takes a maximum value of 40 in our study.
To empirically measure ex post core violation, for each of the $100$ instances, we sample one panel from an algorithm and compute the core violation using~\Cref{alg:auditing}. 
We note that this is not exactly equal to the worst-case core violation, but a very good approximation of it.

\subsection{Results}
\paragraph{Results for Ex Post Core Violation.}
In \adult{} dataset, we observe an unbounded ex post core violation for \unifsel{} when $q \le 4$. Specifically, for $q \in \{1, 2, 3\}$, we observed unbounded core violation in $84\%$, $9\%$, and $36\%$ of the instances respectively. This happens since ${\sim}8.3\%$ of the population is mapped to a single data point and that \unifsel{} fails to select $q$ individuals from this group. When $q \le 3$, we have $\nicefrac{q}{k} \le 8.4\%$, and this cohesive group is entitled to select at least $q$ members of the panel from themselves, which results in $q$-cost of $0$ for them and an unbounded violation of the core.  However, $\fgc$ captures this cohesive group and selects at least $q$ representatives from them. Furthermore, we see significantly higher ex post core violation for \unifsel{} compared to $\fgc$ for smaller values of $q$ (up to $12$) and comparable performance for larger values of $q$. This is expected as $\fgc$ tends to behave more similarly to \unifsel{} as $q$ increases because it selects from fewer yet larger groups ($\floor{\nicefrac{k}{q}} + 1$ groups of size $\nicefrac{qn}{k}$).

We observe a similar pattern in \ess{}-UK that \unifsel{} obtains worse ex post core violations when $q$ is smaller and similar performance as $\fgc$ for larger values of $q$. However, in contrast to \adult{}, we do not observe similar unbounded violations for \unifsel{} in \ess{}-UK. The reason is that \ess{}-UK consists of ~250 features (compared to the $5$ we used from \adult{}) and any data points represent at most $0.2\%$ of the population. Thus, no group is entitled to choose enough representatives from their own to significantly improve their cost or make it $0$. The decline in core violation for $q = k$ happens as it measures the minimum improvement in cost over the whole population, which is more demanding than lower values of $q$.
Lastly, $\fgc$ performs consistently for all values of $q$ and achieves an ex post core violation less than $1.6$ and $1.25$ in $\adult{}$ and \ess{}-UK respectively.

\begin{figure}[t!]
	\centering
	\begin{minipage}{\textwidth}
		\centering
		\begin{subfigure}[t]{0.4\textwidth}
			\includegraphics[width=\textwidth]{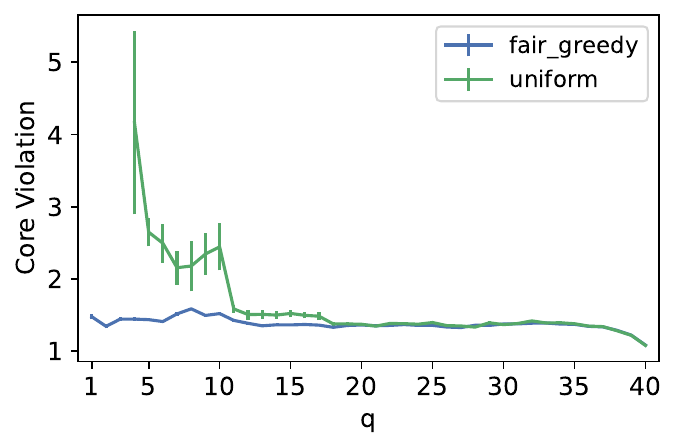}
			\caption{\adult{}}
			\label{fig:adult-ex-post}
		\end{subfigure}
		\hspace{0.0275\textwidth}
		\begin{subfigure}[t]{0.41\textwidth}
			\centering
			\includegraphics[width=\textwidth]{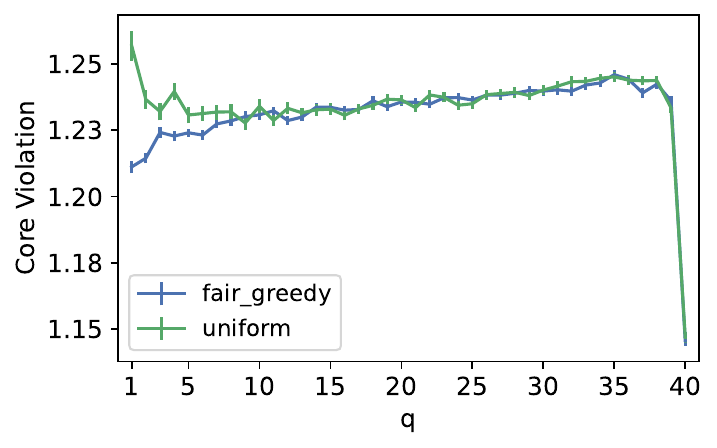}
			\caption{\ess{}-UK}
			\label{fig:ess-ex-post}
		\end{subfigure}
		\caption{Ex post core violation of $\fgc$ and \unifsel{} with $k = 40$}
		\label{fig:ex-post-core}
	\end{minipage}
\end{figure}

\begin{figure}[t!]
    \begin{minipage}{\textwidth}
    \centering
    \begin{subfigure}[t]{0.44\textwidth}
        \includegraphics[width=\textwidth]{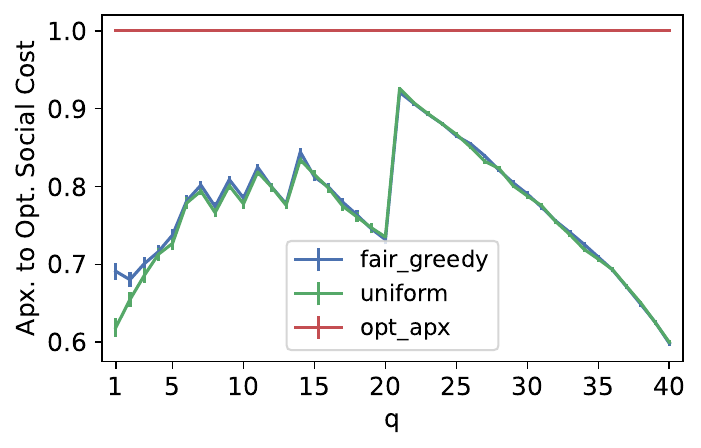}
        \caption{\adult{}}
        \label{fig:adult-repr}
    \end{subfigure}
    \hspace{0.04\textwidth}
    \begin{subfigure}[t]{0.44\textwidth}
        \includegraphics[width=\textwidth]{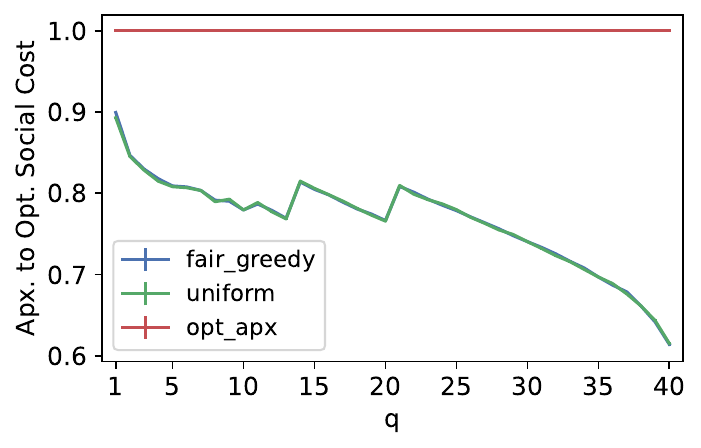}
        \caption{\ess{}-UK}
        \label{fig:ess-repr}
    \end{subfigure}
    \caption{Approximation to the optimal social cost of $\fgc$ and \unifsel{} with with $k = 40$}
    \label{fig:apx-opt-cost}
    \end{minipage}
\end{figure}

\paragraph{Evaluating Approximation to Optimal Social Cost.}
As we mentioned in the introduction, \citet{SKMPS22} use a different approach to measure the representativeness of a panel by considering the social cost (sum of $q$-costs) over a panel. In particular, they define the representativeness of an algorithm  as the worst-case ratio between the optimal social cost  and the (expected) social cost obtained by the algorithm. \citet{SKMPS22}, in their empirical analysis,  measure the average approximation to the optimal social cost of an algorithm $\mathcal{A}$ over a set of instances $\mathcal{I}$, defined as
$\frac{1}{|\mathcal{I}|} \sum\nolimits_{I \in \mathcal{I}} \frac{\min_P ~ \sum_{i \in [n]} c_q(i, P)}{\sum_{i \in [n]} c_q(i, \mathcal{A}(I))}$. Since finding the optimal panel is a hard problem and the dataset and panel sizes are large, \citet{SKMPS22} use a proxy for the minimum social cost, specifically, an implementation of the algorithm of \citet{KR13} for the fault-tolerant $k$-median problem that achieves a constant factor approximation of the optimal objective --- which is equivalent to minimizing the $q$-social cost. 
We use the same approach and report the  average approximation to the optimal social cost of $\fgc$ and \unifsel{}. 

In \Cref{fig:apx-opt-cost}, the reader can see the performance of the two different algorithms over this objective. For \ess{}-UK, we observe a similar behaviour from the two algorithms, while for \adult{}, $\fgc$ outperforms \unifsel{} for $q \in [3]$, which is again due to $\fgc$ capturing the cohesive group. All considered, we observe that $\fgc$ can maintain at least the same level or even better optimal social cost approximation as \unifsel{} would, while achieving significantly better empirical core guarantees in the two datasets.

\section{Discussion}

This work introduces a notion of proportional representation, called the core,  within the context of sortition. The core serves as a metric to ensure proportional representation across intersectional features. While uniform selection achieves an ex ante $O(1)$-$q$-core, it fails to provide a reasonable approximation to the ex post $q$-core. To address this, we propose a selection algorithm, $\fgc$, which preserves the positive aspects of uniform selection, i.e. fairness and ex ante $O(1)$-$q$-core, while also meeting the ex post $O(1)$-$q$-core requirement. We also highlight that the use of $\fgc$ allows the translation of the core requirement into a set of quotas, which can be integrated with another set of quotas to ensure proportional representation across both individual and intersectional features.

It is worth to emphasize that the limitations of uniform selection in satisfying ex post guarantees arise from the potential return of non-proportionally representative panels with a positive probability.  In~\Cref{app:expected-cost}, we explore a natural variation where the core property is mandated to hold over the expected $q$-costs of panels chosen from a selection algorithm. We demonstrate that this variation is incomparable with the ex post $q$-core, and more importantly uniform selection  fails to offer any meaningful multiplicative approximation to this variation, as well.

There are many directions for future work. First, there are gaps between the lower and upper bounds  we provide for both the ex ante and the ex post core.  Closing these gaps and investigating if there are fair selection algorithms that provide better guarantees to ex ante and/or ex post core is an immediate interesting direction. Moreover, we show Fair Greedy Capture is in the ex post $6$-$q$-core, but we do not provide  lower bounds  indicating that  this analysis is tight.   In fact, 
 in~\Cref{fgc-q=1}, we show that for $q=1$, $\fgc$   is in the ex post $((3+\sqrt{17})/2 \approxeq 3.57)$-$1$-core and this is tight. Finding tight bounds for the general case is an open question. In addition, in~\Cref{sec:afgc}, we show that if $q$ is known for an application at hand and we wish to provide guarantees with respect to ex post $q$-core,  a variation of $\afgc$ provides an approximation of  $((5 + \sqrt{41})/2)\approx 5.72$, which is slightly better than the approximation of $6$. Exploring if this is tight as well is another interesting direction. 
Furthermore, \citet{micha2020proportionally} show that for $q=1$, Greedy Capture \citep{chen2019proportionally},  provides better guarantees for the Euclidean space. So, another interesting question is to see if when the metric $d$  consists of usual distance functions such as norms $L^2$, $L^1$ and $L^{\infty}$,  $\fgc$ can provide better guarantees.

\bibliographystyle{plainnat}
\bibliography{abb,bibliography}

\newpage
\appendix

\newpage

\section{Minimizing Social Cost Fails to Provide Proportional Representation}\label{app:social-cost}
\begin{example}
Let $n$ be odd, $k=3$ and $q=1$.  Assume that there are four group  of individuals, $A$, $B$, $C$ and $D$. There are exactly one individual in group $A$, and exactly one individual in group $B$, while there are $\frac{n - 1}{2}$ individuals  in group $C$  and   $\frac{n - 1}{2}$ individuals in group  $D$.  The  distances between individuals in different groups is specified in the following table. 
\begin{table}[ht]
    \centering
    \begin{tabular}{l|ccccc}
         & $A$ & $B$ & $C$ & D \\ 
        \hline 
        $A$ &   $0$  & $\infty$   & $\infty$ & $\infty$ \\
        $B$ &   $\infty$  & $0$   & $\infty$ & $\infty$ \\
        $C$ & $\infty$ & $\infty$  & $0$   & $10$\\
        $D$ &  $\infty$ & $\infty$  & $10$   & $0$\\
    \end{tabular}
\end{table}

\noindent It is not difficult to see that any  panel with minimum social cost contains the single individuals in groups $A$  and $B$ and  one  individual from either group $C$ or group $D$,  as otherwise the social cost would be unbounded. This means  that while  the individuals in group $C$ form almost $50\%$ of the population, and similarly do the individuals  in group $D$, in any  panel with optimal social cost, either group $C$ or $D$  is not represented at all. On the other hand, the two eccentric individuals are always part of the panel. 
\end{example}

\section{Uniform Selection is in the Ex Post $2$-$k$-Core}\label{app:US-ex-post-core}

Next, we show that when $q=k$, any panel is  in the  ex post $2$-$k$-core, which implies that any algorithm including  uniform selection is in the  ex post $2$-$k$-core.   This is due to the fact that in this case only if the grand coalition, i.e. all the agents, has incentives to deviate, the ex post core is violated. 
\begin{theorem}\label{thm:US-ex-post-k=q}
Every panel is in the ex post $2$-$k$-core.  Therefore, uniform selection is in the ex post $2$-core,  and this  is  tight. 
\end{theorem}
\begin{proof}
    Consider any panel $P$. It suffices to show that for any arbitrary panel $P'$ of size $k$, the $q$-cost of all individuals cannot be improved by a factor of greater than $\alpha = 2$. 

    Let $i_1$ and $i_2$ be the two individuals in the population with the maximum distance between them.
    Now, consider an arbitrary representative $r$ in panel $P'$. 
    Without loss of generality, suppose that $c_k(i_1,P')\leq c_k(i_2,P')$.  Then, we have
    \begin{align*}
       c_k(i_2,P) = {\max_{j \in P}} \; d(i_2, j)
       &\leq  d(i_1,i_2) & \text{(by the choice of $i_1$ and $i_2$)}
       \\
       &\leq d(i_1,r) + d(r, i_2) & \text{(triangle inequality)}\\
       &\leq c_k(i_1, P') + c_k(i_2, P') & \text{(as $r \in P'$)}\\
       & \leq 2\cdot  c_k(i_2, P'). 
    \end{align*}
    This implies $\CV(P,P',2)< |P'|\cdot n/k=n$, since the $q$-cost for $i_2$ does not improve by a factor of more than two. From, this we get that any panel $P$ is in the ex post $2$-$k$-core, and therefore uniform selection is in the ex post $2$-core.

    Next, we show that there exists an instance such that uniform selection is not in the ex post $\alpha$-$k$-core for $\alpha<2$. Consider the case that the individuals are assigned into three groups, $A$, $B$ and $C$, with  $\floor{k/2}$,   $\ceil{k/2}$, and $n-k$ individuals, respectively. The  distances between individuals is as specified in the following table. 
    
\begin{table}[ht]
    \centering
    \begin{tabular}{c|ccc}
         & $A$ & $B$ & $C$ \\ 
        \hline
        $A$ &    $0$  & $2$   & $1$ \\
        $B$ &   $2$  & $0$   & $1$\\
        $C$ &    $1$  & $1$   & $0$\\
    \end{tabular}
\end{table}
   
\noindent The panel $P$ which consists of all the $k$ people in groups $A$ and $B$ is in the support of uniform selection. Then, for $i \in A \cup B$, $c_k(i, P) = 2$ as the $k$-th closest representative in $P$ lies in the other group. For $i \in C$,  the $c_k(i, P) = 1$. Now, consider panel $P'$ that consists of $k$ individuals from group $C$. The $q$-costs of all individuals improve by a factor of at least $2$. Hence, $\sort$ violates ex post $2$-$k$-core in this example.
\end{proof}

\section{Uniform Selection is in the Ex Ante \texorpdfstring{$k$}{}-Core} \label{app:US-ex-ante-k}
\begin{proposition}
       Uniform selection is in the ex ante $k$-core.
\end{proposition}
\begin{proof}
    To satisfy ex ante $k$-core,  for any panel $|P'|$ of size $k$, we should have
    \begin{align*}
    \E_{P \sim \calD_k}[\CV(P,P',\alpha)] <  |P'|\cdot \frac{n}{k}=n.
    \end{align*}
    Essentially, this means that the ex ante $k$-core is violated only if the grand coalition, i.e. all the agents, has incentives to deviate to $P'$, in expectation.   
    Since $\CV(P,P',\alpha) \le n$ for all $P'$ by definition, it suffices to show that there exists a panel $P$ that is chosen with non-zero probability, and it holds that $\CV(P,P',\alpha) < n$.  Since,  $\sort$ chooses any panel with non-zero probability, including $P'$, there is a non-zero probability that we realize panel $P = P'$ for which $\CV(P,P',\alpha) = 0$ --- since the $q$-costs do not strictly improve for any individual. Thus, the expected preference count of the panel that selected from uniform selection with respect to any other panel is strictly less than $n$, satisfying the ex ante $k$-core.
\end{proof}

\section{Proof of Chu–Vandermonde Identity }\label{append:CV-identity}
\begin{proof}
We give a combinatorial argument for this identity. Suppose we want to select $k + 1$ items out of a set of size $n + 1$. For $i \in [1, n + 1]$, let $P_{i}$ be the number of such subsets in which the $(r+1)$th picked item is item $i$. As each subset is counted exactly once among $P_i$'s (at the position of its $(r + 1)$th item), we have $\sum_{i = 1}^{n + 1} P_i = \binom{n + 1}{k + 1}$. Now, we calculate $P_i$.
Suppose the $(r+1)$th item is $i$. Then, $r$ items should be selected from the first $i - 1$ items and $k + 1 - (r + 1) = k - r$ items should be selected from the last $n + 1 - i$ items. Therefore,
$P_i = \binom{i - 1}{r} \cdot \binom{n - (i - 1)}{k - r}$. Then, we have
\[
\binom{n + 1}{k + 1} = \sum_{i = 1}^{n + 1} P_i = \sum_{i = 1}^{n + 1} \binom{i - 1}{r} \cdot \binom{n - (i - 1)}{k - r}
= \sum_{j = 0}^{n} \binom{j}{r} \cdot \binom{n - j}{k - r}.
\qedhere
\]
\end{proof}

\section[q-Core over Expected Cost]{$q$-Core over Expected Cost} \label{app:expected-cost}

A variation of the demanding ex post $q$-core is to ask the core property to hold with respect to the expected $q$-cost, as it is given in the definition below. 

\begin{definition}[$\alpha$-$q$-Core over Expected Cost]
    A selection algorithm $\A_{k}$ is in the  $\alpha$-$q$-core over expected cost (or in the $q$-core over expected cost, for $\alpha=1$) if there is no $S\subseteq [n]$ and a panel $P'$ with $|P'|\leq |S|/n \cdot k $  such that
    \begin{align*}
        \forall i\in S, \ \E_{P \sim \A_{k}} [c_q(i,P)]>\alpha \cdot c_q(i,P').
    \end{align*} 
\end{definition}

 We start by showing that the ex post $q$-core and the $q$-core over expected cost are incomparable. 

\begin{proposition}
For any $q\in [k]$, ex post $q$-core and $q$-core over expected cost are incomparable.
\end{proposition}
\begin{proof}
    First, we show that the  ex post $q$-core does not imply the $q$-core over expected cost.  Assume that $n$ is divisible by $k$ and $q$ is divisible by $3$. Consider an instance where there are  five groups of individuals, $A$, $B$, $C$, $D$ and $E$. 
    The first three groups contain $(q\cdot n/k-q)/3$ individuals each, the fourth group contains $q$ individuals and the last group contains $n-q\cdot n/k$ individuals. The table below provides the specified distances between individuals within given groups.
\begin{table}[ht]
    \centering
    \begin{tabular}{c|ccccc}
        & $A$ & $B$ & $C$ & $D$ & $E$ \\ 
        \hline
        $A$ &  $0$  & $2 $   & $2$ & $1$ & $\infty$\\
        $B$ &  $2 $  & $0$   & $2 $ & $1$ & $\infty$\\
        $C$ &  $2 $  & $2 $   & $0$ & $1$ & $\infty$\\
        $D$ &  $1$  & $1$   & $1$ & $0$ & $\infty$\\
        $E$ &  $\infty$  & $\infty$    & $\infty$ & $\infty$ & $0$\\ 
    \end{tabular}
\end{table}

\noindent
    Suppose that a selection algorithm $\A_{k}$  returns with probability $1/3$ a panel that contains $q$ individuals from group $A$ and the remaining representatives are from group $E$, with probability $1/3$ a panel that contains $q$ individuals from group $B$ and the remaining representatives are from group $E$ and with probability $1/3$ a panel that contains $q$ individuals from group $C$ and the remaining representatives are from group $E$. All these panels  are in the ex post $q$-core, since there is no sufficiently large group such  that if they choose another panel, all of them reduce their distance. 
    Now, we see that for each $i$ in $A$ or $B$ or $C$, it holds that 
    \begin{align*}
      \E_{P \sim \A_{k}} [c_q(i,P)]= \frac{2}{3}\cdot 2 =4/3 
    \end{align*}
    while for each $i$ in $D$, it holds that 
     \begin{align*}
      \E_{P \sim \A_{k}} [c_q(i,P)]= 1.
    \end{align*}
    If all the individuals in $A$, $B$, $C$ and $D$ choose a panel $P'$ that contains $q$ individuals from $D$, then all of them reduce their distance by a factor at least equal to $4/3$. 

     Next, we show that the  $q$-core over expected cost does not imply the ex post $q$-core.     
    Consider an instance where there are  four groups of individuals, $A$, $B$, $C$, $D$. 
    Group $A$ contains $q\cdot n/k-q$ individuals,  group $B$ contains $q$ individuals,  group $C$ contains $q$ individuals and  group $D$ contains all the remaining individuals.  The  distance between individuals belonging to given groups is specified in the following table. 
\begin{table}[ht]
    \centering
    \begin{tabular}{c|cccc}
        & $A$ & $B$ & $C$ & $D$ \\ 
        \hline 
        $A$ &   $0$  & $1$   & $2$ & $\infty$ \\
        $B$ & $1$  & $0$   & $1$ & $\infty$\\
        $C$ &  $2$  & $1$   & $0$ & $\infty$\\
        $D$ &  $\infty$  & $\infty$   & $\infty$ & $0$\\
    \end{tabular}
\end{table}

\noindent
Suppose that a selection algorithm $\A_{k}$  returns with probability $1/2$ a panel $P_1$ that contains $q$ individuals from group $A$  and $k-q$ individuals from group $D$, and with the remaining probability returns  a panel $P_2$ that contains $q$ individuals from  group $C$ and $k-q$ individuals from group $D$. Then, for each $i$ in $A\cup C$, we have that 
    \begin{align*}
      \E_{P \sim \A_{k}} [c_q(i,P)]= \frac{1}{2}\cdot 2=1
    \end{align*}
    while for each $i$ in $B$, we have that
     \begin{align*}
      \E_{P \sim \A_{k}} [c_q(i,P)]= \frac{1}{2}  \cdot 1+ \frac{1}{2}  \cdot 1=1.
    \end{align*}
Hence,  this algorithm is in the $q$-core over expected cost.  But when the algorithm returns $P_1$,  all the individuals in $A$ and $B$ can reduce their cost by a factor of $2$ by choosing $q$ representatives in $B$. 
\end{proof}

Next, we show that as in the case of the ex post $q$-core, uniform selection is in the $2$-$k$-core over expected cost. 

\begin{theorem}\label{thm:US-expected-core-k=q}
    For $q=k$, uniform selection is   $2$-$q$-core over expected cost.  
\end{theorem}
\begin{proof}
In the proof of~\Cref{thm:US-ex-post-k=q}, we show that for any $P$ and any panel $P'$, with $|P|=|P'|=k$, there exists $i\in N$, such that $c_k(i,P)\leq 2\cdot  c_k(i,P')$.     This implies that $
       \E_{P\sim \sort }[c_k(i,P)] \leq 2\cdot  c_k(i, P'),
    $
    which means that uniform selection is  in the $2$-$k$-core over expected cost. This is because to violate $2$-$k$-core over expected cost, the $k$-cost of the entire population would have to improve by a factor of more than $2$, which does not hold for individual $i$.
\end{proof}

Again as in the case of the ex post $q$-core, we show that uniform selection does not provide any bounded multiplicative approximation to the $q$-core over expected cost, for $q\in [k-1]$. 
\begin{theorem}
For any $q \in [k-1]$ and  $\floor{\nicefrac{n}{k}}\ge k$,  there exists an instance such that  uniform selection is  not in the  $\alpha$-$q$-core over expected cost, for any bounded $\alpha$. 
\end{theorem}
\begin{proof}
Consider the instance as given in proof ~\Cref{thm:US-impos-ex-post-core}.  As before, uniform selection may return a panel that consists only from individuals in group $A$. Therefore, all the individuals in group $B$ have  positive expected $q$-cost under uniform selection, while if they choose a panel among themselves, they would all have a $q$-cost of $0$. Thus, uniform selection is not in the  $\alpha$-$q$-core over expected cost for any bounded $\alpha$. 
\end{proof}

Lastly, we show that $\fgc_{k}$ is in the $6$-$q$-core over expected cost, for every $q$. 
\begin{theorem}\label{thm:GC-core-over-expected-core}
For every $q$, $\fgc_{k}$  is  in the  $6$-$q$-core over expected cost. 
\end{theorem}

\begin{proof}
Let $\calD_{k}$ be the distribution that $\fgc_k$ returns.   Suppose for contradiction that there exists  $S\subseteq [n]$ and $P'\subseteq [n]$, with $|S|\geq |P'| \cdot n/k$,  such that 
      \begin{align*}
          \forall i \in S,  \quad \quad \E_{P\sim \calD_{k,q}}[c_{q}(i,P)]>6 \cdot  c_{q}(i,P').
      \end{align*} 

In the proof of~\Cref{thm:ex-post-core}, we show that there exist $i^*_{\ell}\in S$ such that  for every $P$ in the support of the algorithm, we have that $c_q(i^*_{\ell},P)\leq 6\cdot c_q(i^*_{\ell},P')$. This implies that   $\E_{P\sim \calD_{k}}c_q(i^*_{\ell},P)\leq 6\cdot c_q(i^*_{\ell},P')$ which is a contradiction. 
\end{proof}

\section{Analysis of $\fgc_{k}$ for the Ex Post $1$-Core}\label{fgc-q=1}
\begin{theorem}\label{theor-greedy-capture-q-1}
   $\fgc_{k}$  in the  ex post $\frac{3+\sqrt{17}}{2}$-$1$-core and there exists an instance for which this bound is tight. 
\end{theorem}
\begin{proof}
 Let $P$ be any panel that the  algorithm may return.   Suppose for contradiction that there exists a panel $P'$  such that $V_q(P,P',(3+\sqrt{17})/2)\geq |P'| \cdot n/k $. This means that there exists $S\subseteq[n]$, with $|S|\geq |P'|\cdot n/k$, such that: 
\begin{align}
{
  \forall i \in S,  \quad \quad c_{q}(i,P)> (3+\sqrt{17})/2 \cdot  c_{q}(i,P').
  }
\end{align}
If  $|P'|>1$,  we can partition $S$ into $|P'|$ groups by assigning each individual to their closest representative  from $P'$, and at least one of these groups  should have size at least $n/k$.   Therefore, without loss of generality, we can assume that $|P'|=1$ and $|S|\geq n/k$. 

Let $P'=\{i^*\}$ and  $i'$ be the individual in $S$ that has the largest distance from $i^*$. 
 Since there are sufficiently many individuals in the ball $B(i^*, d(i^*,i'))$, it is possible that the algorithm opened it during its execution. If this happened, this means that  there is at least one representative in $P$ that is located within this ball. Then, we get that $i'$ has a distance at most equal to the diameter of the ball from her  closest representative in $P$ which is at most $2\cdot d(i', i^*)= 2\cdot c_q(i',P')$. Hence,  $i'$ cannot  reduce her distance by a multiplicative factor larger than $2$ by choosing $P'$, and we reach in a contradiction. 
 
 On the other hand, if the algorithm did not open this ball during its execution, this means that some of the individuals in $T$ have been allocated to different balls before the ball centered at $i^*$ captures sufficiently many of them. Hence,  some individuals in $S$ have been captured from a different ball with radius at most $d(i',i^*)=c_q(i', P')$. Suppose that $i''$ is the first individual in $S$ that was captured from such a ball. Then, we have that within this ball there is $1$ representative in $P$. Hence $c_q(i'',P)\leq 2\cdot d(i',i^*)$, since the distance of $i''$ form any other individual in this ball is at most equal to the diameter of the ball. We consider the minimum multiplicative improvement of both $i'$ and $i''$: 
\begin{align*}
& \min \left(  \frac{c_q(i',P)}{c_q(i',P')}, \frac{c_q(i'',P)}{c_q('',P')} \right)\\
    =& \min \left(  \frac{c_q(i',P)}{d(i',i^*)}, \frac{c_q(i'',P)}{d(i'',i^*)} \right)\\
    \leq & \min \left(  \frac{d(i',i'')+c_q(i'',P)}{d(i',i^*)}, \frac{c_q(i'',P)}{d(i'',i^*)} \right)
    &\text{(by \Cref{lem:cost-distance})}\\
    \leq & \min \left(  \frac{d(i',i^*)+d(i^*,i'')+ c_q(i'',P)}{d(i',i^*)}, \frac{c_q(i'',P)}{d(i'',i^*)} \right)
    \\
    \leq & \min \left(  \frac{d(i',i^*)+d(i^*,i'')+ 2\cdot d(i',i^*)}{d(i',i^*)}, \frac{2\cdot d(i',i^*)}{d(i'',i^*)} \right)
    & \text{(as $c_q(i'',P)\leq 2\cdot d(i',i^*)$ )}
    \\
     \leq &\max_{z\geq 0} \min(3+ 1/z, 2\cdot z) = (3+\sqrt{17})/2.
\end{align*}

To show that this bound is tight consider the case that $n=28$ and $k=7$. Assume that the individuals form 
four isomorphic sets of $7$ individuals each such that  each set is sufficiently far from all other sets.   The distances between the individuals in  one set are given in the table below.
\begin{table}[ht]
    \centering
    \begin{tabular}{c|c ccc ccc}

        & $a_1$ & $a_2$ & $a_3$ & $a_4$ & $a_5$ & $a_6$ & $a_7$\\ 
        \hline  
        $a_1$ &   $0$  & $1$   & $2$ & $\frac{\sqrt{17}-1}{2}$& $\frac{\sqrt{17}+1}{2}-\epsilon$& $\frac{\sqrt{17}+1}{2}-\epsilon$& $\frac{\sqrt{17}+3}{2}-2\cdot \epsilon$ \\

        $a_2$ &   $1$  & $0$   & $1$ & $\frac{\sqrt{17}-3}{2}$& $\frac{\sqrt{17}-1}{2}-\epsilon$& $\frac{\sqrt{17}-1}{2}-\epsilon$& $\frac{\sqrt{17}+1}{2}-2\cdot \epsilon$ \\
       
        $a_3$ &   $2$  & $1$   & $0$ & $\frac{\sqrt{17}-1}{2}$& $\frac{\sqrt{17}+1}{2}-\epsilon$& $\frac{\sqrt{17}+1}{2}-\epsilon$& $\frac{\sqrt{17}+3}{2}-2\cdot \epsilon$ \\
     
        $a_4$ &   $\frac{\sqrt{17}-1}{2}$  & $\frac{\sqrt{17}-3}{2}$   & $\frac{\sqrt{17}-1}{2}$ & $0$& $1-\epsilon$& $1-\epsilon$& $2-2\epsilon$ \\
         
        $a_5$ &   $\frac{\sqrt{17}+1}{2}-\epsilon$  & $\frac{\sqrt{17}+1}{2}-\epsilon$   & $\frac{\sqrt{17}-1}{2}-\epsilon$ & $1-\epsilon$& $0$& $0$& $1-\epsilon$ \\
            
        $a_6$ &   $\frac{\sqrt{17}+1}{2}-\epsilon$  & $\frac{\sqrt{17}+1}{2}-\epsilon$   & $\frac{\sqrt{17}-1}{2}-\epsilon$ & $1-\epsilon$& $0$& $0$& $1-\epsilon$ \\
        
        $a_7$ &   $\frac{\sqrt{17}+3}{2}-2\epsilon$  & $\frac{\sqrt{17}+3}{2}-2\epsilon$   & $\frac{\sqrt{17}+1}{2}-2\epsilon$ & $2-2\epsilon$& $1-\epsilon$& $1-\epsilon$& $10$ \\
  
    \end{tabular}
\end{table}

\noindent Since $k=7$ and there are four isomorphic groups, there exists a group that has at most one representative in some realized  panel. Note that the algorithm first opens the balls that are centered at $a_5$ and have radius equal to $1-\epsilon$. Assume that when this ball  was opened in  the group that has one representative in the panel, the algorithm chooses $a_7$ to be included in the panel. Then, in this group the individuals $a_1$, $a_2$, $a_3$ and $a_4$ are eligible to choose $a_2$ and all of them reduce their distance by a multiplicative factor of at least $(3+\sqrt{17})/2$ as $\epsilon$ goes to zero. 
\end{proof}

\section{$\afgc$ with Known $q$}\label{sec:afgc}

\begin{algorithm}[t] 
	\caption{$\afgc$}\label{alg:GC-q}
	\KwIn{Individuals $[n]$, metric $d$, $k$, $q$}
	\KwOut{Panel $P$}
        
   	$R\gets [n]; \delta\gets 0; P \gets \emptyset$\;
    \While{$|R|\geq \ceil{q \cdot n/k}$}{
        Smoothly increase $\delta$\;
        \While{$\exists j\in R $ such that $|B(j,\delta) \cap R|\geq \ceil{q \cdot n/k }$}{
       $S\gets  \ceil{ q\cdot n/k} $ individuals arbitrary chosen from $B(j,\delta)$\;
        $\hat{P}\gets$ pick $q$ individuals from $S$ uniformly at random\;
         $P\gets P \cup \hat{P}$\;
        $R\gets R \setminus S$\;
        }}
    \If{$|P|<k$}{
       $\hat{P}\gets $  $k-|P|$ individuals from $ [n]\setminus  P$ by picking $i\in R$ with probability $k/n$ and $i\in [n]\setminus (P\cup R)$ with probability $\frac{k-|P|- |R|\cdot k/n }{n  - |P| - |R|}$\;
        $P \gets P \cup \hat{P}$\;}
\end{algorithm}

Here, we show that there exists a version of $\fgc$ such that if $q$ is known, it provides an approximation of $((5 + \sqrt{41})/2)$ to the ex post $q$-core. 
 As before, our algorithm leverages  the basic idea of  the Greedy Capture algorithm.   
 
 $\afgc_{k,q}$, in~\Cref{alg:GC-q},  starts with an empty panel $P$ and  grows a ball around every  individual in $[n]$ at the same rate. When  a ball captures  $\ceil{q\cdot n/k}$ individuals (if more than $\ceil{q\cdot n/k}$ individuals 
 have been captured, it chooses exactly $\ceil{q\cdot n/k}$ by arbitrarily excluding some points on the boundary),  the algorithm selects  $q$ of them uniformly at random, includes them in the panel $P$,  and  disregards all the  $\ceil{q\cdot n/k}$ individuals. When this happens,  we say that the algorithm {\em detects} this ball.  Unlike Greedy Capture, we continue growing balls only around the individuals that are not yet disregarded, i.e. detected balls are frozen. When fewer than $\ceil{q\cdot n/k}$ individuals remain, the algorithm selects the remaining representatives from among the individuals who have not yet been included in the panel as follows: each individual who has not been disregarded is selected with probability $k/n$, and the remaining probability mass is allocated uniformly among the individuals who have been disregarded but not selected. This can be achieved through systematic sampling~\cite{yates1948systematic}.

\begin{theorem}
For every $q$, $\afgc_{k,q}$ is fair and  in the ex post $((5 + \sqrt{41})/2$-$q$-core. 
\end{theorem}
\begin{proof}

We start by showing that the algorithm is fair. 
\begin{lemma}\label{lemma:fairness-of-fgc}
$\afgc_{k,q}$ is fair.  
\end{lemma}
\begin{proof}
     Suppose that $q\cdot n/k$ is an integer. Then, each individual that is disergarded in the while loop of the algorithm is included in the panel with probability exactly $k/n$. Now, suppose that after the algorithm has detected $t$ balls,   less than $q\cdot n/k$ individuals have not been disregarded. Then, when the algorithm exits the while loop we have that $|R|=n-t\cdot q\cdot n/k$ and $k-|P|=k-t\cdot q$. But since, 
\begin{align*}
    |R|\cdot k/n= k-t\cdot q,
\end{align*}
we conclude that  the remaining $k-t\cdot q$ representatives are chosen uniformly among the individuals in $R$. Thus,  the algorithm returns a panel of size $k$ and each $i\in [n]$ is chosen with probability $k/n$.

Now, we focus on the case  where $q\cdot n/k$ is not an integer. In this case,  note that in the while loop of the algorithm, less  than $k$ representatives are included in the panel, since $q$ representatives are included in it every time that $\ceil{q\cdot n/k}$ non-disregarded individuals are captured from a ball. Moreover,  each individual that is disregarded is chosen with probability strictly less than $k/n$. Now suppose that after exiting the while loop, there are individuals that have not been disregarded, i.e. $|R|>0$. First, we show  that the algorithm correctly chooses another $k - |P|$ representatives and outputs a panel of size $k$.  The algorithm would select each individual in $R$ with probability $k/n$ and allocates the remaining probability --- which is equal to $k-|P|- |R|\cdot k/n$ ---uniformly  among the $n - |P| - |R|$ individuals  that have been disregarded but not selected in $P$.
To satisfy fairness for people in $R$, it suffices to show that $|R|\cdot k/n < k-|P|$. Since for each individual $i \in [n] \setminus R$ we have $\Pr[i \in P] = q/\ceil{q\cdot n/k} < k /n$, then we have $|P| = \E[|P|] = \sum_{i \in [n] \setminus R} \Pr[i \in P] < (n - |R|) \cdot k/n$. Thus,
$k-|P|> k-(n-|R|)\cdot k/n=|R|\cdot k/n.$
Hence, the algorithm outputs panels of size $k$.

It remains to show that each  individual in $[n] \setminus R$, which is disregarded in the while loop, is  included in the panel with probability $k/n$. First, note that all of them are included in the panel with the same probability. This holds, since each is selected with probability $q/\ceil{q\cdot n/k}$ from the ball that captured them in the while loop, and, when not selected in the while loop, they get an equal chance of selection of $\frac{k - |P| - |R| \cdot k/n}{n - |R| - |P|}$.  
Since the size of the final panel returned by the algorithm is always $k$, and by linearity of expectation, we have $k = |R| \cdot k/n + \sum_{i \in [n]\setminus R} \Pr[i]$. By equality of $\Pr[i]$'s, we conclude that all must be equal to $k/n$ and each individual in $[n]$ is included in the panel with probability $k/n$.
\end{proof}

 We proceed by showing that $\afgc_{k,q}$ is in the ex post $((5 + \sqrt{41})/2)$-$q$-core.
 Let $P$ be any panel that the  algorithm may return.   Suppose for contradiction that there exists a panel $P'$  such that $\CV(P,P',(5+\sqrt{41})/2)\geq |P'| \cdot n/k $. This means that there exists $S\subseteq[n]$, with $|S|\geq |P'|\cdot n/k$, such that:
\begin{align}\label{eq:afgc-ex-post}
  \forall i \in S,  \quad \quad c_{q}(i,P)> (5+\sqrt{41})/2 \cdot  c_{q}(i,P').
\end{align}

Let $T_1,\ldots, T_m$ be a partition of $S$ with respect to $P'$, as given in the first part of~\Cref{lem:partition}. In a similar way as in the proof of~\Cref{thm:ex-post-core}, we can conclude that there exists a ball centered at $i^*_{\ell}$ that has radius $2\cdot c_{q}(i^*_{\ell},P')$ and captures all the individuals in $T_{\ell}$. Since there are sufficiently many individuals in this ball, it is possible that the algorithm detected this ball (or a nested one) during its execution. If this happened, this means that  there are $q$ representatives in $P$ that are located within the ball $B(i^*_{\ell},2\cdot c_{q}(i^*_{\ell},P'))$. Then, we get that
 $c_{q}(i^*_{\ell},P)\leq 2\cdot c_{q}(i^*_{\ell},P')$ which contradicts \Cref{eq:afgc-ex-post}. 

 On the other hand, if the algorithm did not detect this ball (or a nested one) during its execution, this means that some of the individuals in $T_{\ell}$ have been disregarded before the ball centered at $i^*_{\ell}$ captures sufficiently many of them. Hence,  some individuals in $T_{\ell}$ have been captured from a different ball with radius at most $2\cdot c_{q}(i^*_{\ell},P')$. Suppose that $i'$ is the first individual in $T_{\ell}$ that was captured from such a ball. Then, we get that $q$ representatives  in $P$ are within this ball, which means that  $c_q(i',P)\leq 4\cdot c_{q}(i^*_{\ell},P')$, since the distance of $i'$ form any  individual in this ball is at most equal to the diameter of the ball. We consider the minimum multiplicative improvement of both $i^*_{\ell}$ and $i'$: 
\begin{align*}
    & \min \left(  \frac{c_q(i',P)}{c_q(i',P')}, \frac{c_q(i^*_{\ell},P)}{c_q(i^*_{\ell},P')} \right)\\
    \leq & \min \left(  \frac{c_q(i',P)}{c_q(i',P')}, \frac{d(i^*_{\ell},i')+c_q(i',P)}{c_q(i^*_{\ell},P')} \right)
    &  \text{(by \Cref{lem:cost-distance})}
    \\
    \leq & \min \left(  \frac{c_q(i',P)}{c_q(i',P')}, \frac{d(i^*_{\ell},r_{i'}) +d(i',r_{i'})  +c_q(i',P)}{c_q(i^*_{\ell},P')} \right)  & \text{(by Triangle Inequality)}
    \\
    \leq & \min \left(  \frac{c_q(i',P)}{c_q(i',P')}, \frac{c_q(i^*_{\ell},P') +c_q(i',P')  +c_q(i',P)}{c_q(i^*_{\ell},P')} \right)
    &   \text{(as $r_{i'}\in \top_q(i',P')\cap \top_q(i^*_{\ell},P'))$}\\
     \leq & \min \left(  \frac{4\cdot c_q(i^*_{\ell},P')}{c_q(i',P')}, 5+\frac{ c_q(i',P')  }{c_q(i^*_{\ell},P')} \right) &  \text{(as  $c_q(i',P)\leq 4\cdot  c_q(i^*_{\ell},P'))$}\\
     \leq &\max_{z\geq 0} \min(4\cdot z, 5+1/z) = 
     (5+\sqrt{41})/2
\end{align*}
which  violates \Cref{eq:afgc-ex-post}  and the theorem follows. 
\end{proof}

\end{document}